\pdfoutput=1

\documentclass[a4paper,USenglish,cleveref,autoref,thm-restate]{lipics-v2021}

\hideLIPIcs

\usepackage{iftex}

\usepackage{amsfonts}
\usepackage{algorithm}
\usepackage[noend]{algpseudocode}
\usepackage{booktabs}
\usepackage{float}
\usepackage{verbatim}
\usepackage{stackengine}
\usepackage{quantikz}
\usepackage{braket}
\usepackage{physics}
\usepackage{todonotes}
\usepackage{dsfont}

\usepackage{empheq}
\usepackage{mathrsfs}
\usepackage{dblfloatfix}
\usepackage{cuted}
\usepackage{pgfplots}
\usetikzlibrary{arrows}
\usetikzlibrary{positioning}

\nolinenumbers

\ifPDFTeX
  \usepackage{underscore}
\else
  \usepackage{breakurl}
\fi

\title{Quantum Algorithms for Magic Square Diophantine Equations}

\titlerunning{Quantum Algorithms for Magic Square Diophantine Equations}

\author{Dimitrios Thanos}
  {LIACS, Leiden University}
  {dimitrios.thanos.0@gmail.com}
  {0000-0002-0719-1036}
  {Supported
by the Dutch National Growth Fund, as part of the Quantum Delta NL program.}

\author{Marcello Bonsangue}
  {LIACS, Leiden University}
  {m.m.bonsangue@liacs.leidenuniv.nl}
  {0000-0003-3746-3618}
  {}

\author{Alfons Laarman}
  {LIACS, Leiden University}
  {a.w.laarman@liacs.leidenuniv.nl}
  {0000-0002-2433-4174}
  {}

\authorrunning{D. Thanos, M. Bonsangue, and A. Laarman}

\Copyright{Dimitrios Thanos, Marcello Bonsangue, and Alfons Laarman}



\begin{document}

\maketitle

\begin{abstract}

Magic-square constraints define Diophantine systems whose solutions, in several natural families, exhibit rigid periodic structure. We study this structure in an oracle setting, where a marked set of integers is given by black-box access and the goal is to decide whether it encodes a magic square. For $3\times 3$ magic squares and weighted variants, we prove explicit periodic characterizations that reduce detection to period finding. For larger orders, we identify a class of solutions built from repeated arithmetic patterns, which can be detected via the quantum Fourier transform. We then introduce a shifted-oracle method, based on interference between an oracle and its translates, that helps reconstruct solutions in structured cases. Together, these ingredients give a quantum framework for detecting and reconstructing certain magic-square solutions under suitable assumptions. We also derive finite bounds that make some instances exhaustively solvable and obtain Shor-based criteria for certifying non-existence in restricted number-theoretic settings. As an application, we sketch a quantum communication protocol based on an oracle encoding of a large magic-square solution.

\end{abstract}

\section{Introduction}\label{intro}

Consider a binary string $x \in \{0,1\}^B$. Exactly $n^2$ positions, with $n \in \mathbb{N}$, contain a $1$, and all other positions contain $0$. Define the marked set

$$S=\{\, i \in \{1,\dots,B\}\mid x_i = 1 \,\}.$$

Assume that $x$ is given as a black-box oracle: on query $i \in \{1,\dots,B\}$, the oracle returns the bit $x_i$; see Figure~\ref{fig:memory-bits}. We ask whether the marked set $S$ encodes a solution of a prescribed magic-square\footnote{A magic square of order $n$, i.e. an $n\times n$ magic square.} Diophantine system. Since $|S|=n^2$, this is a decision problem in which the marked set itself is the candidate solution.

\begin{figure}[ht]
\centering
\begin{tikzpicture}[
  bit/.style={draw, minimum width=7mm, minimum height=7mm, inner sep=0pt, font=\ttfamily},
  idx/.style={font=\scriptsize}
]
\def\startindex{0}

\node[bit] (c0) {1};
\foreach \b [count=\i from 1] in {0,1,1,0,0,1,0,1,1,1} {
  \node[bit, right=0pt of c\the\numexpr\i-1\relax] (c\i) {\b};
}
\foreach \b [count=\i from 0] in {0,0,1,1,0,0,1,0,1,1,1} {
  \node[idx, above=2pt of c\i] {\the\numexpr\startindex+\i\relax};
}
\end{tikzpicture}
\caption{Representation of the entries of an oracle encoding a binary string. In the classical setting, only one index can be queried at a time, and the output lies in $\{0,1\}$.}
\label{fig:memory-bits}
\end{figure}

For large $B$, this task may require many oracle queries in the classical setting, since each query reveals only one entry of $x$. If instead the string is provided as a \emph{quantum phase oracle}, the oracle can be queried in superposition, and the Quantum Fourier Transform can be used to extract information about periodic structure in the marked set. When $B=2^q$, a classical FFT-based analysis reads all $B$ oracle values and performs $O(B\log B)$ additional operations. A QFT-based procedure uses one phase-oracle query per shot and $O((\log B)^2)$ elementary quantum gates, followed by repeated sampling and classical verification. 

The connection with magic squares comes from the fact that several solution families exhibit arithmetic-progressions. For $3\times3$ magic squares and weighted variants, we prove explicit periodic characterizations. For larger orders, we identify a class of solutions built from repeated arithmetic progressions; in this class, the number of repetitions scales with the order of the magic square (\autoref{theorem_3_n_progressions}). This repeated structure provides the input pattern exploited by the QFT-based detection step.

Moreover, for sufficiently large order, we also consider a setting in which the oracle may contain additional marked entries. In this setting, the marked set itself is no longer the candidate solution; rather, the task is to detect and reconstruct an \(n^2\)-entry structured subset corresponding to a magic-square solution. When the order is large enough, the repeated
arithmetic-progressions in the solution can contribute enough Fourier weight
to produce prominent peaks, even in the presence of marked points that do not belong to the
solution set. This is compatible with Parseval's theorem: although the Fourier mass is conserved, a large repeated structure can still make the relevant periodicities observable. This is the setting in which a quantum
advantage may arise, since the QFT-based procedure samples the Fourier structure directly,
whereas a classical FFT-based approach must first access the oracle values. In certain structured cases, we device a procedure for the reconstruction of the solutions that seems to be an even stronger candidate for a quantum advantage; It does not only rely on the speedup of QFT vs FFT, but it also employs a shifted-oracle procedure together with the Hadamard test, in a way that, to the best of our knowledge, cannot be matched by a classical algorithm. 

The oracle tasks considered here are deliberately restricted: the goal is not to solve arbitrary magic-square Diophantine systems, but to isolate structured cases in which quantum period-finding and shifted-oracle techniques can be applied. This is in the same spirit as other quantum algorithms and oracle problems that study highly structured settings in order to expose possible quantum advantages, such as Hallgren's algorithm for Pell's equation~\cite{hallgren2007polynomial} and the Forrelation problem~\cite{aaronson2015forrelation}. 

Finally, as an application of our algorithms, we discuss a proof-of-concept quantum communication protocol between two parties with access to a quantum computer.

The paper is organized as follows. \autoref{background} gives background; \autoref{from solutions to periodicities} relates magic-square solutions to periodic patterns; and \autoref{Main alg and large magic sq} presents the QFT-based detection procedure, numerical experiment, and shifted-oracle reconstruction method. The following sections give finite-bound and non-existence results, \autoref{protocol} gives a proof-of-concept communication protocol, and omitted proofs, QFT background, and examples are in the \hyperref[appendix]{Appendix}.

\section{Background}\label{background}

Diophantine equations are polynomial equations with integer coefficients whose solutions are restricted to integers. Matiyasevich’s solution of Hilbert’s tenth problem showed that no general algorithm can determine whether an arbitrary Diophantine equation has an integer solution \cite{hilbert00,matiyasevich70,Jones1980UndecidableDE}. This impossibility motivates specialized approaches for tractable subclasses and the study of their properties in general. Additionally, the study of Diophantine equations is not limited to high-degree or structurally complex cases. That is supported by Skolem's proof that every Diophantine equation is equivalent to one of degree at most four \cite{skolem1934diophantine}, showing that low-degree instances already capture the general problem and therefore deserve dedicated algorithmic treatment. Indeed, one of the landmark quantum algorithms in this area is Hallgren's algorithm for Pell's equation~\cite{hallgren2007polynomial}, a quadratic Diophantine equation of the form $x^2 - dy^2 = 1$, where $d$ is a positive non-square integer. Surprisingly, the classical version of this problem is believed to be harder than integer factorization~\cite{hallgren2007polynomial} while the relevant quantum algorithm can find its solutions in quantum poly-time.

In general, Diophantine equations take the form
$P(x_1,x_2,\ldots,x_n)=0$,
where $P$ is a polynomial with integer coefficients and solutions in $\mathbb{Z}^n$. They range from linear equations to nonlinear problems of major historical significance, such as Fermat's Last Theorem, which has no nontrivial integer solutions for $n>2$, and the Cannonball problem, which admits a unique nontrivial solution \cite{Anglin01021990}. Both involve deceptively simple equations that, famously, took decades to centuries to resolve. Diophantine equations also arise beyond number theory, with connections to logic, scheduling, and integer programming \cite{marker02,buss98,10.1007/978-3-642-15775-2_29,c00e3ffa-4ee2-3869-b578-fa821ca888fa}.

In this work, we focus on magic squares, whose construction can be formulated as a Diophantine system of equations. A \textit{magic square} is a square grid of distinct integers arranged so that the sums of the entries in each row, column, and diagonal are equal. For example, an $n \times n$ magic square can be constructed using the consecutive integers $1,2,\ldots,n^2$. In that case, the common sum is $M = \frac{n(n^2 + 1)}{2}$, known as the \textit{magic constant}. Magic squares constructed in this way are also called \emph{normal}.

From a mathematical perspective, constructing a magic square amounts to solving a system of linear Diophantine equations: requiring each row, column, and diagonal to sum to $M$ imposes linear constraints on the entries. These constraints may be written as
$$
\begin{gathered}
\sum_{j=1}^{n} x_{ij}=M \qquad \text{for all } i=1,\dots,n,\\
\sum_{i=1}^{n} x_{ij}=M \qquad \text{for all } j=1,\dots,n,\\
\sum_{i=1}^{n} x_{ii}=M,\\
\sum_{i=1}^{n} x_{i,n+1-i}=M.
\end{gathered}
$$
where $x_{ij}$ denotes the entry in row $i$ and column $j$. For $n=3$, this gives a system of eight equations.


Magic squares are also related to constraint satisfaction problems since their construction involves assigning values to variables (the cells of the square) subject to constraints (row, column, and diagonal sums being equal to the magic constant)~\cite{rossi06}. Moreover, magic squares provide insights into the solvability of certain classes of Diophantine equations, as their construction involves solving systems of equations with integer solutions \cite{andrews75,dickson19}.

\section{From solutions to periodicities}\label{from solutions to periodicities}

This section analyzes the geometric and algebraic structure of the Diophantine system associated with an $n \times n$ magic square. Exploiting the connection between the periodic structure of magic squares and the corresponding solution space is central to the design of a quantum algorithm for the exact $n^2$-marked decision setting outlined in \autoref{intro} and also deciding the existence or absence of solutions for this type of Diophantine system.

Consider a $3 \times 3$ magic square with integer entries labeled $a$ through $i$, as shown in \autoref{fig:3by3_letters}. Requiring each row, column, and diagonal to sum to the same integer $M$ gives a system of eight linear Diophantine equations in nine variables, together with the standard constraint that all entries are distinct. A solution is thus a 9-tuple of pairwise distinct integers satisfying these equations.

\begin{figure}[ht]
  \centering
  \begin{subfigure}[t]{0.45\linewidth}
    \centering
    \renewcommand{\arraystretch}{1.4}
    \begin{tabular}{|>{\centering\arraybackslash}c|>{\centering\arraybackslash}c|>{\centering\arraybackslash}c|}
      \cline{1-3}
      $a$ & $b$ & $c$ \\
      \cline{1-3}
      $d$ & $e$ & $f$ \\
      \cline{1-3}
      $g$ & $h$ & $i$ \\
      \cline{1-3}
    \end{tabular}
  \end{subfigure}
\hspace{0.01\linewidth}
  \begin{subfigure}[ht]{0.45\linewidth}
    \centering
    
    \begin{equation}\label{dioph_system}
    \left\{
    \begin{aligned}
        x_a + x_b + x_c &= M \\
        x_d + x_e + x_f &= M \\
        x_g + x_h + x_i &= M \\
        x_a + x_d + x_g &= M \\
        x_b + x_e + x_h &= M \\
        x_c + x_f + x_i &= M \\
        x_a + x_e + x_i &= M \\
        x_c + x_e + x_g &= M
    \end{aligned}
    \right.
    \end{equation}
    \label{eq:dioph_system}
  \end{subfigure}

  \caption{An order $3$ magic square and its associated Diophantine system. The entries $a$--$i$ denote the positions in the square, while $x_a,\ldots,x_i$ are the corresponding variables. A valid magic square is a solution of the system, and $M$ is the common sum of each row, column, and diagonal. }
  \label{fig:3by3_letters}
\end{figure}

Since the order of the variables in each equation is irrelevant, each triplet in the system can be viewed as a point in three-dimensional space. All such points lie on the plane $x+y+z=M$, which over $\mathbb{R}$ is a plane in $\mathbb{R}^3$, while over $\mathbb{Z}$, the solutions form a discrete subset of this plane.

Assume that the Diophantine system has a solution. This determines eight triplets, one for each equation. By commutativity, each triplet gives rise to $3! = 6$ equivalent permutations. For example, the first row of \autoref{fig:3by3_letters} corresponds to the six points $(a,b,c)$, $(a,c,b)$, $(b,c,a)$, $(c,b,a)$, $(c,a,b)$, and $(b,a,c)$. The same holds for every row, column, and diagonal of the magic square.

To analyze the geometry of these triplets, we develop a sequence of lemmas describing the structure of the solution set. This yields a geometric proof of \autoref{pattern_theorem}, corresponding to property $P1$ in \cite{Robertson01101996}, where it was originally proved algebraically. Unlike that algebraic approach, however, the geometric framework extends naturally to \autoref{theorem_weighted}, which treats more general $3 \times 3$ magic-square systems with weighted entries, that is, systems in which the variables are multiplied by fixed constants. The strategy remains to detect periodic patterns corresponding to solutions of the generalized system within a prescribed marked set, i.e., the marked set, which is a subset of candidate integers.

We begin with basic definitions and lemmas that fix notation and reveal the geometry of the solution set. Let $\mathcal{L}^{[x,k]}$ denote the intersection of the plane $x+y+z=M$ with the plane $x=k$, where $k\in\mathbb{N}$. The lines $\mathcal{L}^{[y,k]}$ and $\mathcal{L}^{[z,k]}$ are defined analogously. The next lemma states that two points sharing a coordinate lie on the same such line.

\begin{lemma}\label{two_points_common_line}

Two points $(\alpha,\beta,\gamma)$ and $(\delta,\varepsilon,\zeta)$ lie on the same line $\mathcal{L}^{[x,k]}$ if and only if $\alpha=\delta=k$. Analogously:

\begin{itemize}

\renewcommand\labelitemi{$\bullet$}

\item $\beta=\varepsilon=k$ if and only if the points lie on $\mathcal{L}^{[y,k]}$,

\item $\gamma=\zeta=k$ if and only if the points lie on $\mathcal{L}^{[z,k]}$.

\end{itemize}

\end{lemma}

\begin{proof}

Follows directly from the definition of $\mathcal{L}^{[j,k]}$.

\end{proof}

Two lines in the plane $x+y+z=M$ with the same fixed coordinate, such as $\mathcal{L}^{[x,k_1]}$ and $\mathcal{L}^{[x,k_2]}$, will be called parallel; they are also parallel in the usual geometric sense. The next lemma shows that pairs of points sharing a common coordinate determine lines that belong to one of three families of parallel lines.

\begin{lemma}\label{parallels}

Let $p_1=(\alpha,\beta,\gamma)$, $p_2=(\delta,\varepsilon,\zeta)$, $p_3=(\eta,\theta,\iota)$, and $p_4=(\kappa,\lambda,\mu)$ be points in the plane $x+y+z=M$. If $\alpha=\delta=k_1$ and $\eta=\kappa=k_2$, then the pairs $(p_1,p_2)$ and $(p_3,p_4)$ lie on the parallel lines $\mathcal{L}^{[x,k_1]}$ and $\mathcal{L}^{[x,k_2]}$, respectively. Similarly:

\begin{itemize}

\renewcommand\labelitemi{$\bullet$}

\item if $\beta=\varepsilon=k_1$ and $\theta=\lambda=k_2$, then $(p_1,p_2)$ and $(p_3,p_4)$ lie on $\mathcal{L}^{[y,k_1]}$ and $\mathcal{L}^{[y,k_2]}$,

\item if $\gamma=\zeta=k_1$ and $\iota=\mu=k_2$, then $(p_1,p_2)$ and $(p_3,p_4)$ lie on $\mathcal{L}^{[z,k_1]}$ and $\mathcal{L}^{[z,k_2]}$.

\end{itemize}

\end{lemma}

\begin{proof}

This is an immediate consequence of \autoref{two_points_common_line}.

\end{proof}

\subsection{Solutions are hexagons}

For $i,j \in \{1,2,3\}$, let $\mathcal{T}_{(i,j)}$ denote the transposition operator that swaps coordinates $i$ and $j$ of each point in a set. For example,
$\mathcal{T}_{(1,2)}\big( \{ (4,7,3), (5,2,9) \} \big)=\{ (7,4,3), (2,5,9) \}$.

The points determined by the entries of a $3\times3$ magic square have a hexagonal structure. By studying coordinate transpositions, we show that these points form regular hexagons in the plane $x+y+z=M$.

\begin{lemma}\label{triangles}
If the lines $\mathcal{L}^{[x,k_1]}$, $\mathcal{L}^{[y,k_2]}$, and $\mathcal{L}^{[z,k_3]}$ do not intersect at a single point, then their pairwise intersections form an equilateral triangle. This symmetry is responsible for the hexagonal structure of the solution set.
\end{lemma}

We now show that every solution of the magic-square system contains a hexagon. More precisely, from any solution, one can select seven triplets forming a regular hexagon together with its barycenter. These triplets are a subset of all those associated with the solution and need not include one representative from each row, column, and diagonal. One such choice is
$$(g,a,d),\ (f,e,d),\ (f,i,c),\ (g,i,h),\ (b,e,h),\ (b,a,c),\ (g,e,c).$$
These points form a regular hexagon in the plane $x+y+z=M$, with barycenter at $(g,e,c)$.

The argument is illustrated in \autoref{fig:hexagon}. Points lying on the green, red, and blue lines share their first, second, and third coordinates, respectively, as described in \autoref{parallels}. The three colors, therefore, represent the three families of parallel lines. The hexagon is then constructed triangle by triangle using \autoref{parallels} and \autoref{triangles}.

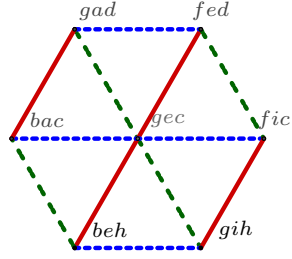
\begin{figure}[ht]
\centering
\resizebox{0.3\linewidth}{!}{%
\begin{tikzpicture}[line cap=round,line join=round,>=triangle 45,x=1cm,y=1cm,scale=0.3]
\definecolor{qqqqff}{rgb}{0,0,1}
\definecolor{ccqqqq}{rgb}{0.8,0,0}
\definecolor{qqwuqq}{rgb}{0,0.39215686274509803}
\definecolor{wwwwww}{rgb}{0.4,0.4,0.4}
\definecolor{uuuuuu}{rgb}{0.26666666666666666,0.26666666666666666,0.26666666666666666}

\clip(-5.58, 2.74) rectangle (6.62, 13.10);

\draw [line width=1.5pt,dash pattern=on 1pt off 1pt on 1pt off 4pt,color=qqwuqq] (-1.9798932449803204,11.995338979376367)-- (3.0201067550196794,3.335084941531976);
\draw [line width=1.5pt,color=ccqqqq] (3.02010675501968,11.995338979376367)-- (-1.979893244980322,3.335084941531976);
\draw [line width=1.5pt,dotted,color=qqqqff] (-4.479893244980323,7.665211960454174)-- (5.520106755019681,7.665211960454171);
\draw [line width=1.5pt,dotted,color=qqqqff] (-1.9798932449803204,11.995338979376367)-- (3.02010675501968,11.995338979376367);
\draw [line width=1.5pt,dash pattern=on 1pt off 1pt on 1pt off 4pt,color=qqwuqq] (3.02010675501968,11.995338979376367)-- (5.520106755019681,7.665211960454171);
\draw [line width=1.5pt,color=ccqqqq] (5.520106755019681,7.665211960454171)-- (3.0201067550196794,3.335084941531976);
\draw [line width=1.5pt,dotted,color=qqqqff] (-1.979893244980322,3.335084941531976)-- (3.0201067550196794,3.335084941531976);
\draw [line width=1.5pt,dash pattern=on 1pt off 1pt on 1pt off 4pt,color=qqwuqq] (-1.979893244980322,3.335084941531976)-- (-4.479893244980323,7.665211960454174);
\draw [line width=1.5pt,color=ccqqqq] (-4.479893244980323,7.665211960454174)-- (-1.9798932449803204,11.995338979376367);

\begin{scriptsize}
\draw [fill=black] (-1.979893244980322,3.335084941531976) circle (2pt);
\draw[color=black] (-.5784847934616513,4.047948103272335) node {$beh$};
\draw [fill=black] (3.0201067550196794,3.335084941531976) circle (2pt);
\draw[color=black] (4.4079753862802176,4.047948103272335) node {$gih$};
\draw [fill=uuuuuu] (5.520106755019681,7.665211960454171) circle (2pt);
\draw[color=uuuuuu] (5.913487397776133,8.395748358515261) node {$fic$};
\draw [fill=uuuuuu] (3.02010675501968,11.995338979376367) circle (2pt);
\draw[color=uuuuuu] (3.4079753862802176,12.718984770508227) node {$fed$};
\draw [fill=uuuuuu] (-1.9798932449803204,11.995338979376367) circle (2pt);
\draw[color=uuuuuu] (-1.0784847934616513,12.718984770508227) node {$gad$};
\draw [fill=uuuuuu] (-4.479893244980323,7.665211960454174) circle (2pt);
\draw[color=uuuuuu] (-3.0839968049575655,8.395748358515261) node {$bac$};
\draw [fill=wwwwww] (0.5201067550196792,7.665211960454172) circle (2pt);
\draw[color=wwwwww] (1.7024633747843032,8.395748358515261) node {$gec$};
\end{scriptsize}
\end{tikzpicture}
}
\caption{Hexagonal diagram summarizing the proof that the chosen points form a hexagon. The labels indicate the corresponding triplets.}
\label{fig:hexagon}
\end{figure}

\begin{proposition}\label{proposition1}
Let $p_1,\ldots,p_6$ be the vertices of a regular hexagon with barycenter $p_0$. For any $i,j\in\{1,2,3\}$, the image
$\mathcal{T}_{(i,j)}(\{p_1,\ldots,p_6\})=\{p'_1,\ldots,p'_6\}$
is also a regular hexagon, with barycenter
$\mathcal{T}_{(i,j)}(\{p_0\})=\{p'_0\}$.
\end{proposition}

Also, we can already formulate a necessary condition for the existence of a solution to the Diophantine system associated with a $3\times3$ magic square: the magic constant $M$ must be divisible by $3$. 

\begin{lemma}\label{middle=n/3}
In the magic square of \autoref{fig:3by3_letters}, the middle entry $e$ must equal $\frac{M}{3}$, where $M$ is the magic constant.
\end{lemma}

\subsection{A hexagon of hexagons}

Let $\mathscr{T}$ denote the set of all transposition operators $\mathcal{T}_{(i,j)}$. The trajectory of a triplet under $\mathscr{T}$ is the set of points obtained by applying all operators in $\mathscr{T}$ to that triplet.

\begin{lemma}\label{trace_of_diag}
The trajectory under $\mathscr{T}$ of a triplet corresponding to a diagonal of the magic square in \autoref{fig:3by3_letters} forms a hexagon.
\end{lemma}

By \autoref{trace_of_diag}, applying the transposition operators to a point $(\alpha,\beta,\gamma)$ whose coordinates correspond to the entries of a diagonal of the magic square produces a hexagon. By \autoref{proposition1}, the same therefore holds for the barycenters of the associated hexagons.

More precisely, \autoref{proposition1} associates a hexagon with any point $p$ whose coordinates are given by the entries of a diagonal of the magic square. Starting from the hexagon of \autoref{proposition1}, one applies a suitable transposition operator $\mathcal{T}_{(i,j)}$ so that the barycenter of the resulting hexagon coincides with $p$.

By \autoref{two_points_common_line} and \autoref{parallels}, these hexagons preserve the same orientation in the plane, as illustrated in \autoref{fig:hyperhex}.

\begin{figure}[ht]
    \centering
    \includegraphics[width=0.25\linewidth]{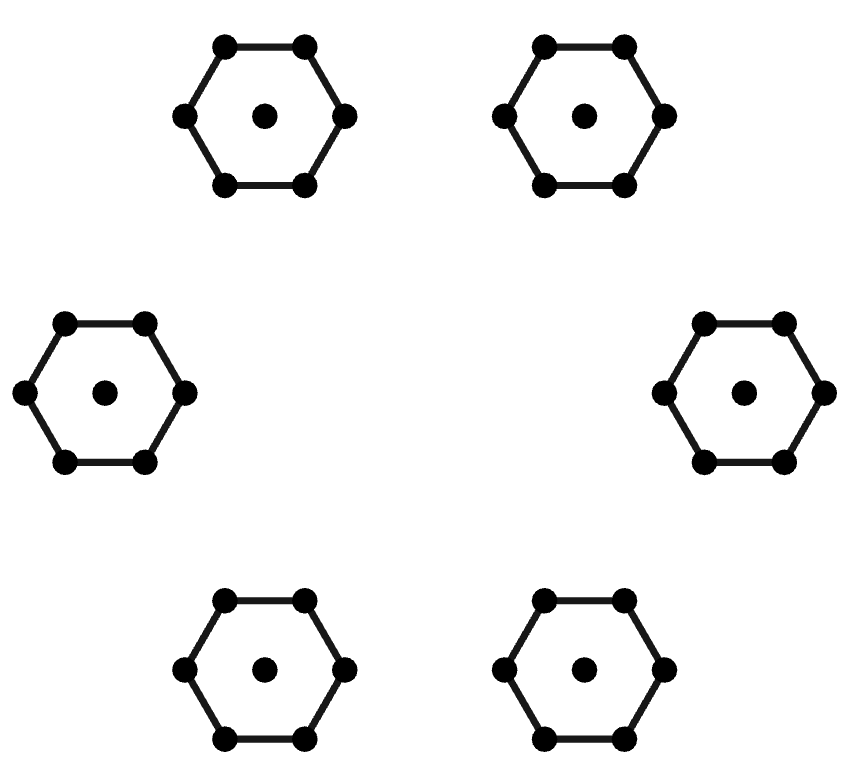}
    \caption{The points obtained by applying coordinate transpositions to the set of points $p_1,p_2,\ldots,p_7$ forming a hexagon with barycenter in the plane $x+y+z=M$, as in \autoref{fig:hexagon}.}
    \label{fig:hyperhex}
\end{figure}

\subsection{The missing points}

So far, we have analyzed $42$ of the $48$ points arising from the $8$ triplets of a $3\times3$ magic square, as summarized in \autoref{fig:hyperhex}. Indeed, each of the $8$ triplets gives rise to $3!=6$ permutations, yielding $48$ points in total. We now show that the remaining $6$ points are not arbitrary points of the plane $x+y+z=M$, but are uniquely determined by the $42$ points already identified.

For every triplet coming from a row, column, or diagonal of \autoref{fig:3by3_letters}, there is another triplet sharing its $j$th coordinate, for each $j\in\{1,2,3\}$, while differing in the other two coordinates. Therefore, when the $48$ points are connected by lines joining points that share a coordinate, each valid point lies on exactly three such lines, by \autoref{two_points_common_line}.

Applying this to the configuration of \autoref{fig:hyperhex} reveals the locations of the remaining $6$ points, shown in \autoref{fig:pattern all lines}. Hence, these points are uniquely determined by the $42$ points already analyzed.

If one considers all intersections of triples of such lines, one obtains $7$ points rather than $6$. One of them must therefore be invalid. The central point, highlighted in red in \autoref{fig:pattern all lines}, is invalid because all three of its coordinates are equal. This also follows from \autoref{trace_of_diag}, since the red point is the barycenter of the corresponding trajectory.

\begin{figure}[ht]
    \centering
    \includegraphics[width=0.3\linewidth]{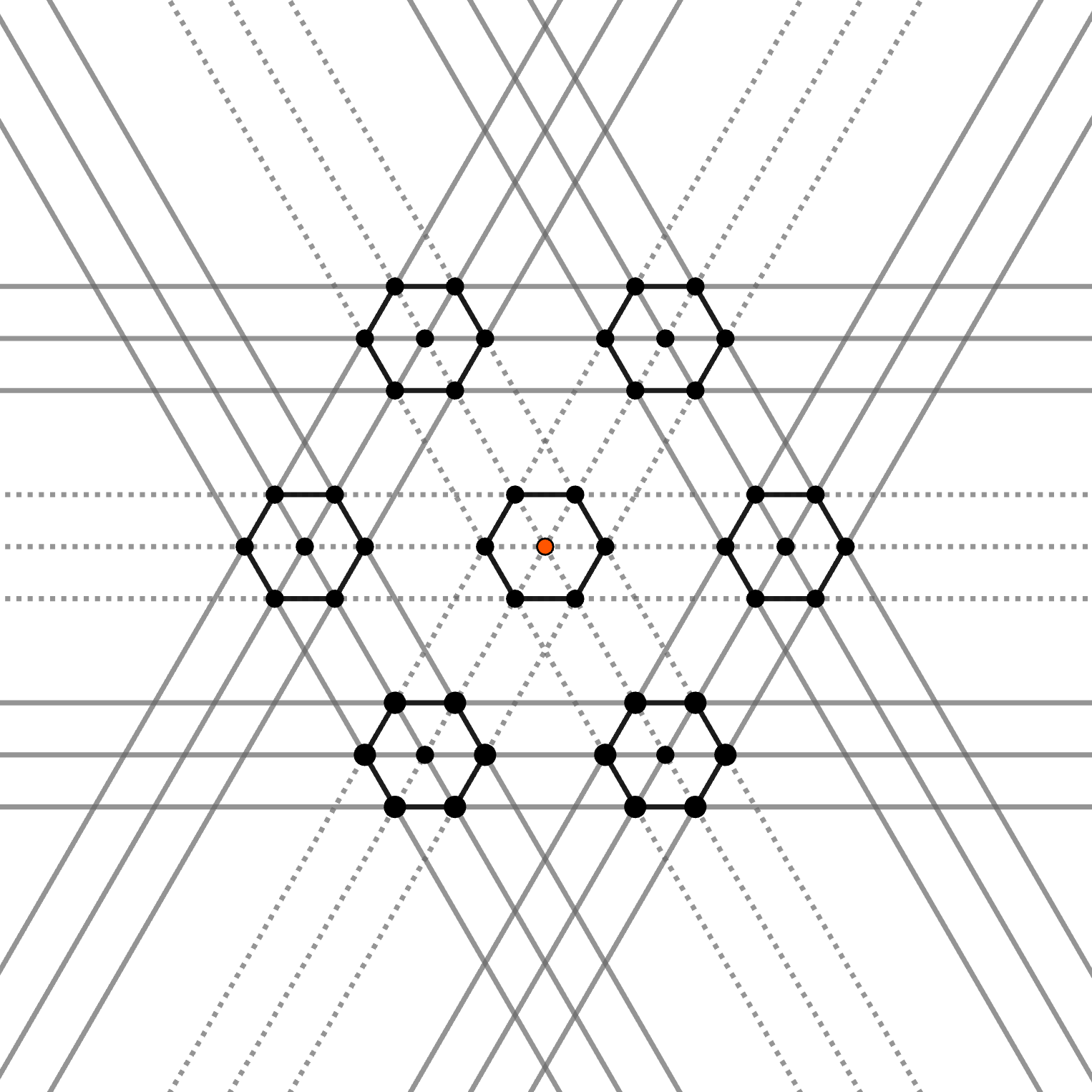}
    \caption{The full configuration of the $48$ triplets associated with \autoref{fig:3by3_letters} is shown by the black points. The central small hexagon consists of the remaining $6$ points, whose locations are uniquely determined by the $42$ points displayed in \autoref{fig:hyperhex}. The red point is invalid, since its coordinates are all equal.}
    \label{fig:pattern all lines}
\end{figure}

\subsection{Solutions form a periodic pattern}\label{sol_to_per}

The next lemma, \autoref{folklore_hexagon}, is a standard geometric fact (see, e.g., \cite{CL91}) that allows us to project points from the plane $x+y+z=M$ onto the planes $x=0$, $y=0$, and $z=0$. These projections, which we call \emph{shadows}, reveal the geometry of the configuration on the coordinate planes. Projecting further onto the coordinate axes then relates the geometry in the plane $x+y+z=M$ to arithmetic patterns on the axes. Since the points arise from a solution of \autoref{fig:3by3_letters}, their shadows have distinct integer coordinates.

\begin{lemma}[folklore]\label{folklore_hexagon}
A cube can be intersected by a plane passing through its barycenter and the midpoints of six edges. This intersection is a regular hexagon. If one vertex of the cube is placed at the origin and its three incident edges lie along the positive $x$-, $y$-, and $z$-axes, then the intersecting plane is parallel to the family of planes $x+y+z=n$, where $n\in\mathbb{N}$.
\end{lemma}

\autoref{fig:hexagon_in_cube} illustrates this fact: the marked edge midpoints lie in a plane whose intersection with the cube is a regular hexagon.

\begin{figure}[ht]
    \centering
    \includegraphics[width=0.3\linewidth]{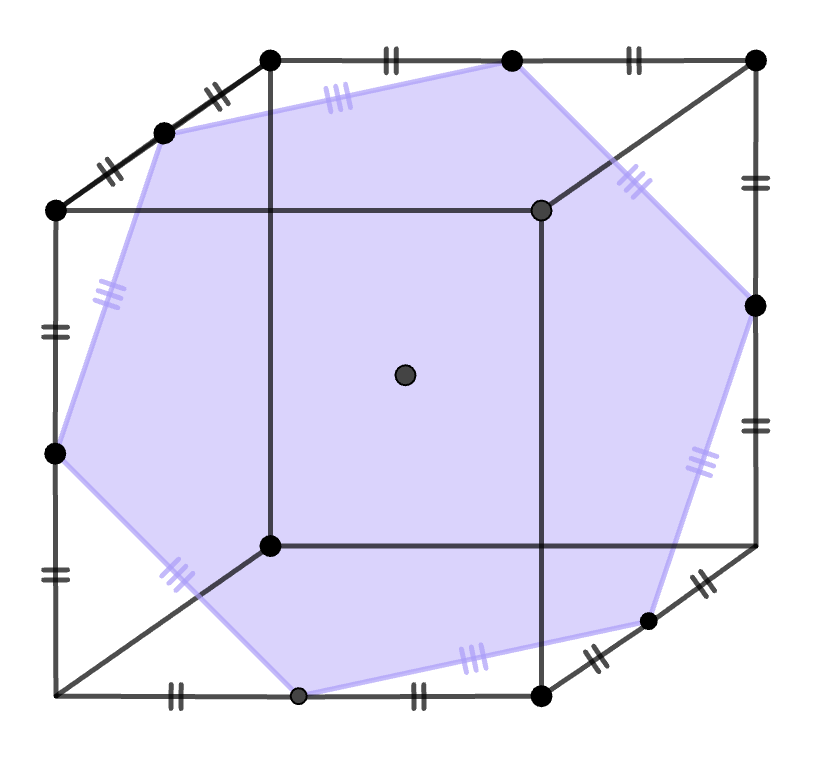}
    \caption{A cube intersected by a plane whose cross-section is a regular hexagon. If one vertex of the cube is at the origin $(0,0,0)$ and the three incident edges lie along the positive $x$-, $y$-, and $z$-axes, then the slicing plane has the form $x+y+z=r$, with $r\in\mathbb{R}$.}
    \label{fig:hexagon_in_cube}
\end{figure}

\begin{corollary}\label{corollary_cube_projection}
Projecting the vertices of the hexagon in \autoref{folklore_hexagon}, together with its barycenter, onto a face of the cube produces two opposite vertices of a square, the midpoints of its sides, and the square's barycenter. Projecting this shadow further onto an edge of the cube yields three points in arithmetic progression.
\end{corollary}

\autoref{fig:cube_projection} provides a visual proof. Note that the two black points on the bottom face of the cube are not projections of the vertices of the hexagon.

\begin{figure}[ht]
  \centering
  \begin{subfigure}[b]{0.3\linewidth}
    \centering
    \includegraphics[width=\linewidth]{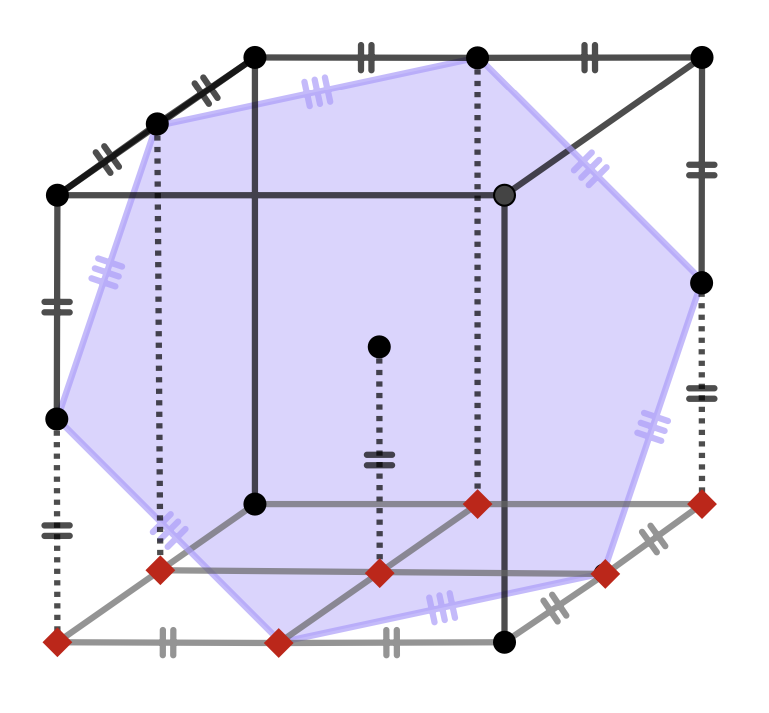}
  \end{subfigure}
  \hspace{2.4cm}
  \begin{subfigure}[b]{0.3\linewidth}
    \centering
    \raisebox{0.3cm}{
    \includegraphics[width=\linewidth]{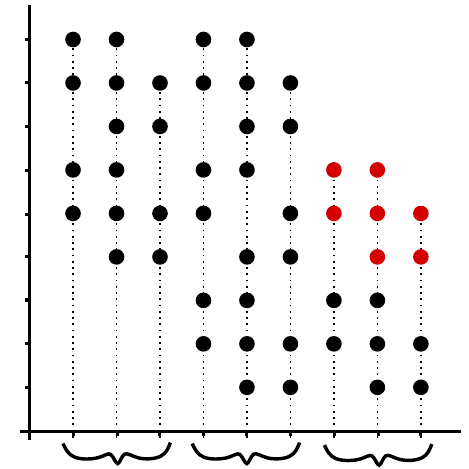}
    }
  \end{subfigure}
  \caption{%
    \emph{Left panel:} The red points on the bottom face of the cube represent the \emph{shadow}, namely the projections of the vertices and barycenter of the hexagon.
    \emph{Right panel:} Projection of all valid points in \autoref{fig:pattern all lines} onto one of the coordinate planes. The braces indicate the arithmetic progressions obtained by projecting further onto an axis. The red points correspond to the projection of a selected hexagon.
  }
  \label{fig:cube_projection}
\end{figure}

We can now identify the periodic structure associated with a solution of the Diophantine system arising from a $3\times3$ magic square: the solution is encoded by arithmetic progressions. This observation reduces the search for integer solutions to the detection of a specific periodic pattern. The next theorem makes this precise.

\begin{theorem}[P1,\cite{Robertson01101996}]\label{pattern_theorem}
Finding the entries of a $3\times3$ magic square is equivalent to identifying a periodic pattern of $9$ integers in the marked set. These integers must form three arithmetic progressions of length $3$, all with common difference $k\in\mathbb{N}$. Moreover, the middle terms of these three progressions must themselves form an arithmetic progression with common difference $K\in\mathbb{N}$.
\end{theorem}

\subsection{Generalization to powers and weighted variables}

The techniques used to prove \autoref{pattern_theorem} also extend to Diophantine systems with weighted variables. In place of the plane $x+y+z=M$, one considers a weighted plane such as $2x+y+z=M$. Although the geometric argument is similar, it is then more natural to work directly with the weighted Diophantine system, since the equations need not correspond to fixed weights on the entries of \autoref{fig:3by3_letters}. Nevertheless, the same framework still identifies solutions through periodic structure.

Because a weighted plane, such as $2x+y+z=M$, is tilted relative to $x+y+z=M$, its projection onto the coordinate axes exhibits different periodic patterns. The structure remains periodic, but may appear stretched or squeezed depending on the weights. For proof, see \autoref{proofs}.

\begin{theorem}\label{theorem_weighted}
Consider a Diophantine system analogous to \autoref{dioph_system}, but with weighted variables, so that the equations have the form $w_xx+w_yy+w_zz=M$, where $w_x,w_y,w_z\in\mathbb{N}$ are fixed constants. Identifying a solution to this system is equivalent to finding integers in a periodic pattern analogous to that of \autoref{pattern_theorem}, except that the pattern may appear several times in stretched or squeezed form, depending on the weights:
\begin{itemize}
    \item it appears \textbf{three times} if $w_i\neq w_j$ for all $i,j\in\{x,y,z\}$;
    \item it appears \textbf{twice} if $w_i=w_j$ for some $i,j\in\{x,y,z\}$, but $w_i\neq w_k$ for the remaining variable $k$;
    \item it appears \textbf{once} if $w_x=w_y=w_z$, in which case the system reduces to \autoref{pattern_theorem}.
\end{itemize}
The pattern is contracted by factors of $\frac{w_c}{w_a}$ and $\frac{w_c}{w_b}$ relative to its most stretched occurrence, where $w_c$ is the largest weight and $w_a,w_b$ are the other two.
\end{theorem}

Finally, note that in the proof of \autoref{pattern_theorem}, the argument depends only on the entries being integers. It therefore applies equally when the marked set is restricted to integers of a prescribed form, such as squares or higher powers. For example, a magic square of squared entries is obtained by searching for the same periodic pattern within the perfect squares. The same reasoning extends to higher powers, as in \autoref{fig:overview_squared}, and more generally to cases where different entries are restricted to different power sets.

\begin{figure}[ht]
  \centering
  \begin{subfigure}[ht]{0.45\linewidth}
    \centering
    \renewcommand{\arraystretch}{1.5}
    \begin{tabular}{|>{\centering\arraybackslash}c|>{\centering\arraybackslash}c|>{\centering\arraybackslash}c|}
      \cline{1-3}
      $a^z$ & $b^z$ & $c^z$ \\
      \cline{1-3}
      $d^z$ & $e^z$ & $f^z$ \\
      \cline{1-3}
      $g^z$ & $h^z$ & $i^z$ \\
      \cline{1-3}
    \end{tabular}
  \end{subfigure}
\hspace{0.01\linewidth}
  \begin{subfigure}[ht]{0.45\linewidth}
    \centering
    $
    \left\{
    \begin{aligned}
      x_a^z + x_b^z + x_c^z &= M \\
      x_d^z + x_e^z + x_f^z &= M \\
      x_g^z + x_h^z + x_i^z &= M \\
      x_a^z + x_d^z + x_g^z &= M \\
      x_b^z + x_e^z + x_h^z &= M \\
      x_c^z + x_f^z + x_i^z &= M \\
      x_a^z + x_e^z + x_i^z &= M \\
      x_c^z + x_e^z + x_g^z &= M
    \end{aligned}
    \right.
    $
    \label{eq:dioph_system_squared}
  \end{subfigure}
  \caption{A magic square of $z$th powers and its associated Diophantine system.}
  \label{fig:overview_squared}
\end{figure}

\section{Main algorithms and large magic squares}\label{Main alg and large magic sq}

\subsection{QFT-based detection and numerical experiment}

We first present the QFT-based detection step, which extracts candidate periodicities from the marked set. In the setting of \autoref{theorem_3_n_progressions}, this identifies candidate values for the common difference of the arithmetic progressions. If representative entries are known, these candidates can be used for reconstruction; otherwise, an additional recovery step is needed.

Longer periodicities improve QFT performance: under reasonable assumptions, the long periodicities make the algorithm more robust and facilitate the separation and identification of periodic patterns~\cite{nielsen00}. In particular, we may search for solutions to large magic squares under additional constraints, for example, requiring all entries to be $z$th powers, i.e.\ $x_i=j^z$ for $i\in{1,\ldots,n^2}$ and $j\in\mathbb{N}$. In that scenario, the Diophantine system is not only the linear constraints of the magic square, but also any extra constraints. As in the exact $n^2$-marked decision setting, a bound $B$ specifies the largest integer included in the marked set. By Parseval's theorem \cite{stein2003fourier}, the sum of the squared magnitudes in the time domain of the Fourier transform equals the sum of squared magnitudes in the frequency domain. Hence, for a sufficiently large magic square, if a solution of the form described in \autoref{theorem_3_n_progressions} exists, the repeated arithmetic-progressions structure is likely to interfere constructively and produce prominent peaks in the Fourier spectrum. This is, however, a relative criterion: a sufficiently strong periodic component must appear among the more prominent peaks, but reliable detection is guaranteed only when there are not too many other strong periodicities present. The prominent periodicities will appear with indices quantified in \autoref{sec:searchQFT}.

\autoref{theorem_3_n_progressions} links periodic patterns in the marked set to solutions of magic squares of arbitrary order. As the order $n$ increases, the associated arithmetic progressions become longer and more numerous.

\begin{theorem}\label{theorem_3_n_progressions}
For a magic square of any order $n>3$, with $n\neq 6$, a solution can be constructed whenever the marked set contains a pattern of $n^2$ integers arranged into $n$ arithmetic progressions of length $n$ with a common difference.
\end{theorem}

The statement of \autoref{theorem_3_n_progressions} is motivated by the connection between arithmetic progressions and solutions of order-$3$ magic squares discussed in \autoref{sol_to_per}. Its proof, however, is substantially different and extends beyond the order 3 case. As shown in \autoref{proofs}, it proceeds by constructing magic squares from arithmetic progressions using tools from combinatorial design, in particular orthogonal Latin squares \cite{raghavarao2005block,c8ff8faf0b5e412fb9dcf5933f672e20,hilton1973double,Brown1993}. In \autoref{examples}, there are examples of order 4 magic squares constructed that way.


The algorithm for the exact $n^2$-marked decision setting and the algorithm for detecting $n^2$-entry patterns inside larger marked sets follow essentially the same QFT-based detection procedure. In both settings, the QFT is used to detect the relevant periodicities. The reconstruction step is then carried out using \autoref{pattern_theorem}, \autoref{theorem_weighted}, or \autoref{theorem_3_n_progressions}, depending on the problem setting. In the setting of \autoref{theorem_3_n_progressions}, once the relevant arithmetic progressions have been recovered, the proof of the theorem gives their explicit placement in the grid via the corresponding diagonal Latin-square construction.

In the case of \autoref{theorem_3_n_progressions}, recovering the starting points of the arithmetic progressions generally requires additional information. One can assume that a representative entry from each progression is given. Under this assumption, Algorithm~\ref{alg:search_qft} gives a complete detection-and-reconstruction pipeline. If these representatives are not given, the shifted-oracle procedure of the next subsection can be used under additional structural assumptions.

\begin{figure}[t]
    \centering
    \begin{quantikz}
    \lstick{$\ket{0}^{\otimes q}$}
        & \gate{H^{\otimes q}}
        & \gate{U_f}
        & \gate{QFT}
        & \meter{} 
    \end{quantikz}
    \caption{Quantum circuit for detecting periodicities in a binary string encoded as a quantum phase oracle $U_f$, where $U_f\ket{x}=(-1)^{f(x)}\ket{x}$. The single wire represents the $q$-qubit computational register. The oracle $U_f$ can either encode the binary string of the exact $n^2$-marked decision setting or encode a search space that might contain solutions to a large magic square (see \autoref{sec:searchQFT})}.
    \label{fig:alg_circuit}
\end{figure}
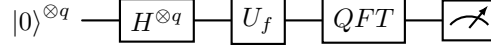

\begin{algorithm}[t]
\caption{QFT-Based Detection of Magic-Square Solutions}
\label{alg:search_qft}
\begin{algorithmic}[1]
\Require A marked set $S \subset \mathbb{Z}$ encoded by a Boolean function $f(x)$; auxiliary reconstruction data for the chosen setting, such as one representative entry from each progression in the setting of \autoref{theorem_3_n_progressions}.

\Ensure Detect/reconstruct a magic-square solution of a prescribed form, if one is found.
\State \textbf{Initialize:} Prepare $q$ qubits in $\ket{0}^{\otimes q}$.
\State \textbf{Apply Hadamard gates:} Create a uniform superposition over all computational basis states.
\State \textbf{Apply phase oracle $U_f$:} Apply $U_f\ket{x}=(-1)^{f(x)}\ket{x}$.
\State \textbf{Apply QFT:} Perform the Quantum Fourier Transform on the computational qubits.
\State \textbf{Measure:} Record prominent Fourier peaks and extract candidate periodicities.\label{alg1:measure}
\State \textbf{Reconstruct:} Use the appropriate theorem
(\autoref{pattern_theorem} for $n=3$, \autoref{theorem_weighted} for the weighted $3\times3$ case, or \autoref{theorem_3_n_progressions} for $n>3$, $n\neq 6$) to obtain a candidate solution.
\State \textbf{Verify:} Check classically that the reconstructed entries are marked by $f$ and satisfy the required Diophantine constraints.
\If{verification succeeds}
    \State \Return the solution.
\Else
    \State \Return ``No solution of the prescribed form found.''
\EndIf
\end{algorithmic}
\end{algorithm}

A measured Fourier outcome $r$ provides information about the common difference through an approximation of the form $\frac{r}{Q} \approx \frac{a}{k}$ for some integer $a$. In that case, the continued-fraction algorithm, as used in Shor's period-finding procedure \cite{shor1994algorithms}, can recover a candidate denominator $k$ from the sample. Repeating the QFT yields a list of candidate denominators, and the subsequent classical verification step checks whether one of them gives rise to the family of $n$ arithmetic progressions required by \autoref{theorem_3_n_progressions}.

As a numerical illustration, we considered marked sets consisting of $13$ arithmetic progressions of length $13$, with additional noise chosen so that $50\%$ of the inputs are $1$. \autoref{QFT_Shots} compares the exact spectrum with the outcomes from $40$ measurement shots for a binary string containing $13$ arithmetic progressions of length $13$, with common difference $5$ and shifts $68i$, $i=1,\ldots,13$. Because $68 \bmod 5 \neq 0$, these shifts do not produce maximal constructive interference. Nevertheless, the prominent peaks remain identifiable as outcomes with higher probability.

\begin{figure}[t]
    \centering
    \includegraphics[width=0.8\linewidth]{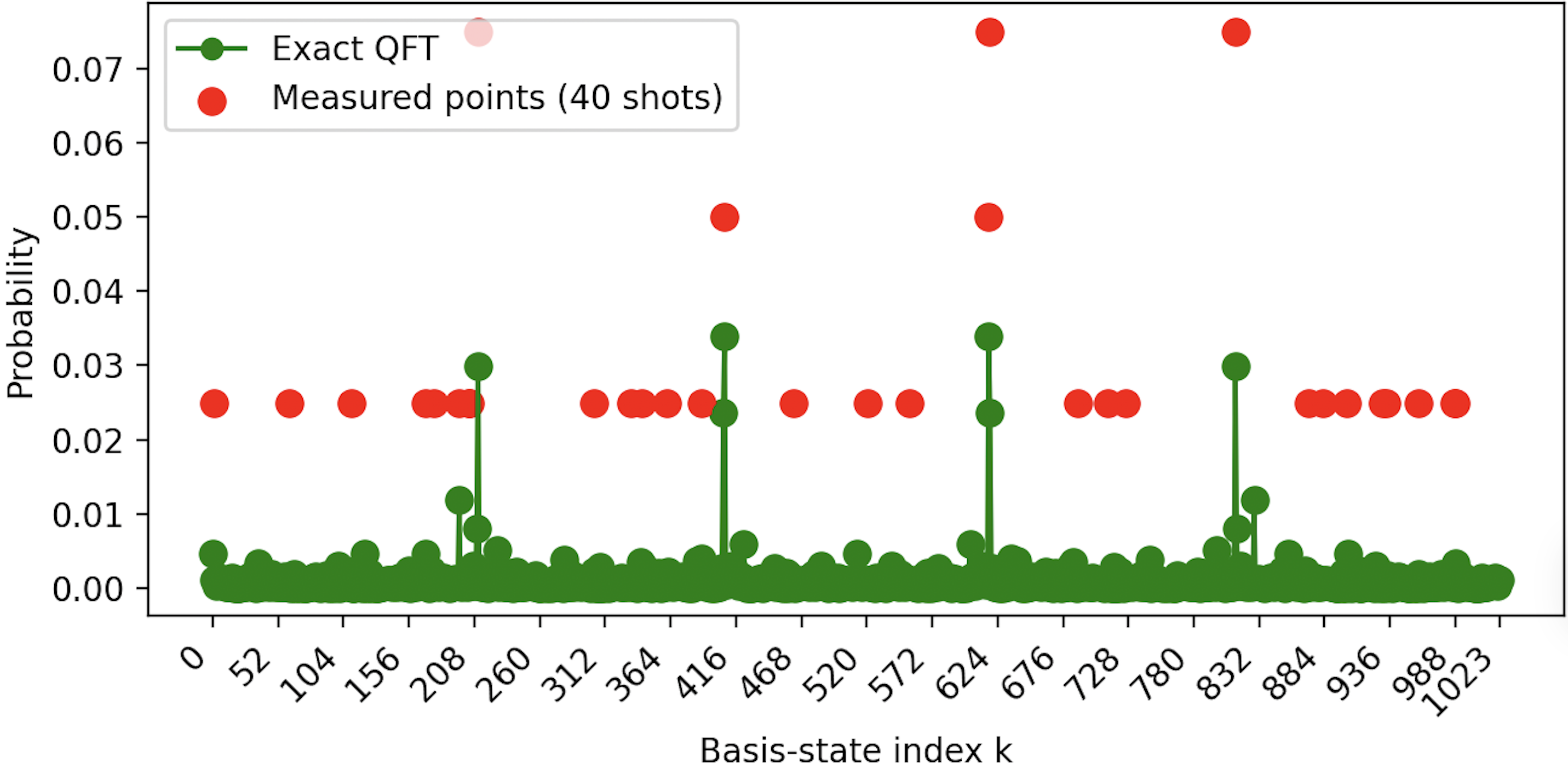}
    \caption{Exact QFT and 40-shot estimate for a binary string containing 13 arithmetic progressions of length 13, with common difference $5$ and shifts $68i$, $i=1,\ldots,13$, with additional random noise so that $50\%$ of the entries are $1$s. The prominent peaks are identifiable as shots with high probability.}\label{QFT_Shots}
    \label{fig:spectrum1}
\end{figure}

This completes the QFT-based detection stage. We next consider how to recover the progression starting points when the representative entries are not given.

\subsection{Recovering starting points with shifted oracles}

Suppose now that none of the representative entries of the arithmetic progressions are given in advance. Then Algorithm~\ref{alg:search_qft} can still be run up to line~\ref{alg1:measure} to identify a candidate common difference $k$, but it does not by itself determine the starting points $l_1,\dots,l_n$. If the solution further satisfies the hypotheses of \autoref{Recovering_D_shift}, these starting points can be recovered by the shifted-oracle procedure described below; see \autoref{appendix_shifted_oracles} for the construction of the shifted oracle.
 
The use of a quantum phase oracle is essential here. Because the oracle can be queried in superposition, the shifted oracle $U_f^{(s)}=T_s^\dagger U_f T_s$ can be interfered with the original oracle $U_f$ and the resulting overlap signal can be estimated directly by a Hadamard-test \cite{lin2022lecturenotesquantumalgorithms,10.1007/978-3-030-35423-7_21}. By contrast, evaluating the analogous correlation for a given shift $s$ classically would generally require explicit access to many entries of the binary string, making the shift-and-compare step inefficient.

\begin{algorithm}[t]
\caption{Recovering Structured Starting Points via Shifted Oracles and Hadamard-Tests}
\label{alg:shifted_oracle_recovery}
\begin{algorithmic}[1]
\Require A phase oracle $U_f$ marking a union of $n$ arithmetic progressions of length $n$; a known common difference $k$.
\Ensure Recover the spacing $D$ between progression starting points and use it to reconstruct a candidate solution.

\State Prepare $\ket{\psi}=B^{-1/2}\sum_{x=1}^{B}\ket{x}$ over the oracle domain.
\For{each tested shift $s$}
    \State Construct the shifted oracle $U_f^{(s)} = T_s^\dagger U_f T_s$.
    \State Interfere $U_f$ and $U_f^{(s)}$ on $\ket{\psi}$ using a Hadamard-test circuit.
    \State Estimate the corresponding overlap signal for shift $s$.
\EndFor

\State \textbf{Identify peaks:} From the estimated overlap values, locate the off-center peaks of the autocorrelation-type signal.
\State \textbf{Recover spacing:} Infer the spacing $D$ from the centers of these peaks.
\State \textbf{Construct candidate square:} Use \autoref{theorem_3_n_progressions} to construct the corresponding magic square.
\State \textbf{Verify:} Classically verify that the reconstructed entries are marked by $f$ and satisfy the required Diophantine constraints.

\If{verification succeeds}
    \State \Return the solution.
\Else
    \State \Return ``No structured solution found.''
\EndIf
\end{algorithmic}
\end{algorithm}

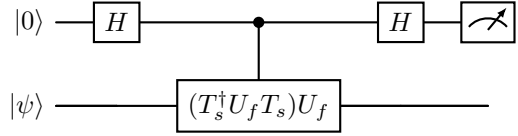
\begin{figure}[t]
\centering
\begin{quantikz}
\lstick{$\ket{0}$} & \gate{H} & \ctrl{1}  & \gate{H} & \meter{} \\
\lstick{$\ket{\psi}$} & \qw & \gate{(T_s^\dagger U_f T_s)U_f}  & \qw & \qw
\end{quantikz}
\caption{Hadamard-test circuit used in \textbf{Algorithm 2}. The controlled operation composes the shifted oracle $U_f^{(s)} = T_s^\dagger U_f T_s$ with the original oracle $U_f$. Measuring the control qubit estimates an overlap signal between the original and shifted marked sets. $\ket{\psi}$ is the uniform superposition over the oracle domain.}
\label{fig:shifted_oracle_circuit}
\end{figure}

For the circuit of Figure~\ref{fig:shifted_oracle_circuit}, the expectation value of the control-qubit observable $Z$ is given by $\langle Z\rangle=\Re\!\bigl(\langle \psi \mid U_f U_f^{(s)} \mid \psi \rangle\bigr).$

The assumptions of \autoref{Recovering_D_shift} below are promise assumptions for the recovery guarantee.

\begin{theorem}\label{Recovering_D_shift}
Assume the marked set consists of a union of $n$ disjoint arithmetic progressions of length $n$, all with known step $k$, together with additional marked points. If the progression starting points are equally spaced and have spacing $D>2(n-1)k$, and if the additive perturbation on the autocorrelation is bounded by $|\Delta(s)|\le \eta$ for all relevant shifts $s$, with $\eta$ smaller than half of the one-step drop along the off-center peaks of the clean autocorrelation, then the starting-point pattern can be recovered in quantum polynomial time.
\end{theorem}

\autoref{Recovering_D_shift} establishes the correctness of the shifted-oracle procedure of \textbf{Algorithm 2} under its stated hypotheses. The task addressed by \textbf{Algorithm 2} is to recover the starting points of the arithmetic progressions once \textbf{Algorithm 1} up to line \autoref{alg1:measure} has identified a candidate common difference $k$. In the setting of \autoref{theorem_3_n_progressions}, the marked set contains several arithmetic progressions with the same common difference $k$ but different first terms. Writing these progressions as
$$
A_j=\{l_j,\ l_j+k,\ l_j+2k,\dots,l_j+(n-1)k\}, \qquad j=1,\dots,n,
$$
the QFT step identifies candidate values of $k$, while the remaining difficulty is to determine the first terms $l_1,\dots,l_n$.

To do so, one shifts the oracle and compares it with the original one; see \autoref{appendix_shifted_oracles}. If $U_f^{(s)}$ denotes the oracle shifted by $s$, then the shifted-oracle procedure probes an autocorrelation-type signal whose enhanced contributions occur at shifts of the form
$$
(l_j-l_i)+mk,
$$
where $i,j\in\{1,\dots,n\}$ and $m\in\{-(n-1),\dots,n-1\}$.
Hence, once $k$ is known, the procedure can in principle reveal the difference set
$$
\{\,l_j-l_i : i,j\in\{1,\dots,n\}\,\},
$$
thereby reducing the classical reconstruction problem.

In the structured setting of \autoref{Recovering_D_shift}, the starting points are equally spaced, so the autocorrelation has triangular peaks whose centers determine the spacing $D$ between the arithmetic progressions. This spacing is then used to reconstruct a candidate pattern of first terms and, through \autoref{theorem_3_n_progressions}, the corresponding candidate magic square.

For each fixed shift $s$, the number of Hadamard-test shots depends only on the desired accuracy.

\section{Finite bounds for restricted magic-square systems}

We now turn from detecting structured solutions to bounding restricted search spaces. For some restricted $3\times3$ magic-square systems, we derive finite bounds beyond which no solutions can exist; such bounds make the remaining search exhaustive, either by listing all solutions or proving that none occur.

\begin{proposition}\label{bound_non_existence}
To decide whether there exists a $3\times 3$ magic square with the additional constraint that its largest entry is a power of $z$ while all other entries are powers of $z-1$, it suffices to check finitely many candidates: there is an upper bound $U$ such that any such magic square must have all entries $\le U$.
\end{proposition}

This result extends to magic squares whose entries involve any mixture of two or more powers. The assumption that the largest entry is the one with the different power can also be removed and replaced by the same condition on any other entry. We state the result in this form only because it makes the proof easier to follow. The same type of argument also extends to weighted Diophantine systems, as in \autoref{theorem_weighted}, with terms raised to mixed powers. Our goal here is not to exhaust this class of results, but to demonstrate their existence. The proof of \autoref{bound_non_existence} (see \autoref{proofs}) illustrates how such bounds can be derived for magic squares involving mixed powers. For weighted systems, the weights can be chosen so as to produce arbitrarily large bounds.

\section{Quantum certification of absence of solutions}\label{absence}

We next consider a different type of absence result, based on number-theoretic obstructions. Building on classical number-theoretic results and Shor's quantum algorithm for integer factorization, we can construct quantum tests to certify the absence of solutions in certain restricted settings. Although these settings are specialized, they can be reduced to problems for which quantum algorithms offer an exponential speedup over the best-known classical methods. 

The tests below use \autoref{pattern_theorem} to reduce the $3\times3$ case to the detection of a finite arithmetic pattern, and \autoref{theorem_3_n_progressions} to restrict the larger-order case to solutions built from arithmetic progressions.

By \autoref{pattern_theorem}, a solution to a $3\times3$ magic square exists if and only if the marked set contains a specific pattern of $9$ integers. This pattern is determined by three integers $l$, $k$, and $K$, where:
\begin{itemize}
\item $l$ is the smallest integer in the pattern,
\item $k$ is the common difference of the three arithmetic progressions, and
\item $K$ is the common difference of the arithmetic progression formed by their first terms.
\end{itemize}

This leads to an efficient quantum test for the absence of such solutions, based on Shor's factorization algorithm \cite{Shor_poly,shor1994algorithms}, under the assumption that at least one of the defining distances is a square. Here, the defining distances are $k$ and $K$.

\begin{proposition}\label{byShors}
There exists an efficient quantum algorithm that can certify the absence of solutions corresponding to $3\times3$ magic squares of squares in the marked set, provided that at least one defining distance is a square.
\end{proposition}

Similarly, \autoref{theorem_3_n_progressions} shows that solutions for magic squares of order $n\neq 3,6$ are determined by $n+1$ integers, namely $l_1,l_2,\ldots,l_n,k$, where:
\begin{itemize}
\item $l_1,l_2,\ldots,l_n$ are the first terms of the $n$ arithmetic progressions, and
\item $k$ is their common difference.
\end{itemize}
In this setting, $k$ and $K$ are the defining distances of a solution, whereas $l$ and $l_i$ ($i\in[n]$) are shifting constants. We restrict attention to $n=4,5$ and $n>6$, since these cases exhibit regular structure in their solution spaces, whereas $n=6$ is exceptional and introduces irregularities that complicate both classical and quantum approaches~\cite{andrews2004magic}.

\begin{proposition}\label{byShors2}
For magic squares of squares of any order $n\neq 3,6$ and a given marked set, there exists an efficient quantum algorithm that can certify the absence of solutions of the form described in \autoref{theorem_3_n_progressions}, provided that one defining distance is a square.
\end{proposition}

The algorithm combines Shor's algorithm with number-theoretic constraints arising from the sum of two squares theorem \cite{dudley1969}. For more details, see the proofs in \autoref{proofs}.

\section{Quantum communication protocol}\label{protocol}

As a proof of concept, we apply Algorithm~\ref{alg:search_qft} to a two-party protocol. Party~A has a quantum computer, while Party~B holds a quantum oracle for a secret large magic square; see \autoref{fig:distributed_protocol}. The parties run Algorithm~\ref{alg:search_qft} distributively so that A reconstructs the square. Party~B also adds \emph{small-bias noise} \cite{smallbias}, whose small nonzero Fourier coefficients make guessing harder while minimally affecting the target spectrum.

The protocol can be summarized as follows.

\begin{enumerate}
    \item Party~B constructs a large magic square that has entries of the class described in \autoref{theorem_3_n_progressions}.

    \item Party~B sends a compact classical description specifying one representative entry from each of the arithmetic progressions.

    \item The quantum oracle $U_f$ of Algorithm~1 with small-bias noise is constructed and held by Party~B. 
    
    \item Party~A and Party~B run \textbf{Algorithm~1} distributively: Party~A prepares the state and applies the initial Hadamard transform, Party~B evaluates $U_f$ and returns the output state. Party~A completes the QFT and measurement.
    
    \item Step~4 is repeated until Party~A recovers the periodic structure detected by Algorithm~1, which together with Step~2 establishes the shared magic square.
    
    \item Party~B sends classical instructions pointing to selected entries of the magic square and transformations to be performed on them. Party~A decodes the bits of the message by applying the agreed transformations.
\end{enumerate}


\begin{figure}[t]
\centering
\begin{tikzpicture}[>=latex, line width=0.8pt, font=\small]

\node[draw, rounded corners=2pt, minimum width=2.2cm, minimum height=1.5cm] (A) at (0,0) {};
\node[draw, rounded corners=2pt, minimum width=2.2cm, minimum height=1.5cm] (B) at (4.6,0) {};

\node at (0,0.35) {Party~A};
\node at (4.6,0.35) {Party~B};

\node[draw, rounded corners=2pt, minimum width=0.9cm, minimum height=0.45cm] (QA) at (0,-0.35) {QC};
\node[draw, rounded corners=2pt, minimum width=0.9cm, minimum height=0.45cm] (O) at (4.6,-0.35) {\(U_f\)};

\draw[<->] (A.east) -- node[above, font=\scriptsize] {quantum channel} (B.west);
\draw[->] ($(B.west)+(0,-0.55)$) -- node[above, font=\scriptsize] {classical channel} ($(A.east)+(0,-0.55)$);

\end{tikzpicture}
\caption{Illustrative representation of the protocol. Party~B holds the oracle \(U_f\), while Party~A has a full quantum computer.}
\label{fig:distributed_protocol}
\end{figure}
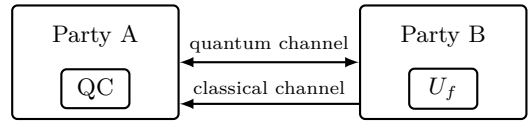

\section{Future work}

Future work includes quantum communication protocols based on large-order magic squares and possible reductions to constrained Latin-square problems, obtained by partially reversing the Latin-square construction used in the proof of \autoref{theorem_3_n_progressions}. This is relevant because Latin squares are closely connected to uniform sampling \cite{088ea90b-206e,8740-320eb5d0f4b6} and can improve on random Monte Carlo sampling.


\nocite{*}
\bibliography{generic}

\clearpage
\appendix
\phantomsection
\label{appendix}
\renewcommand{\sectionautorefname}{Appendix}

\section{Basics of Quantum Fourier Transform}\label{basics_QFT}
QFT is the quantum analog of the classical discrete Fourier transform, and it can be viewed as a change of basis. This transformation maps quantum states in the computational basis into a superposition of basis states in the Fourier basis, defined as:

\begin{equation}\label{QFT}
QFT \ket{x} = \frac{1}{\sqrt{2^n}} \sum_{k=0}^{2^n - 1} \omega^{xk} \ket{k},
\end{equation}

where $x \in \{0, 1, \dots, 2^n - 1\}$ is the integer index of the state in decimal, $\omega = e^{2\pi i / 2^n}$ is the $2^n$-th root of unity, and $k \in \{0, 1, \dots, 2^n - 1\}$ is the index of the Fourier basis state. When the QFT is applied to a superposition of basis states, the phase factor $\omega^{xk} = e^{\frac{2\pi i xk}{2^n}}$ introduces constructive and destructive interference, which is crucial for detecting periodicities.

For example, starting with a state 

\begin{equation*}
\ket{\phi_0} = \sum_{x=0}^{2^n - 1} \alpha_x \ket{x},
\end{equation*}

where $\alpha_x$ are complex coefficients, the QFT maps it to:

\begin{equation*}
QFT \ket{\phi_0} = \sum_{x=0}^{2^n - 1} \alpha_x QFT \ket{x}.
\end{equation*}

Substituting the QFT definition from \autoref{QFT}, we get:

\begin{equation*}
\ket{\phi_1} = \sum_{x=0}^{2^n - 1} \alpha_x \left( \frac{1}{\sqrt{2^n}} \sum_{k=0}^{2^n - 1} \omega^{xk} \ket{k} \right),
\end{equation*}

which can be rearranged as:

\begin{equation*}
\ket{\phi_1} = \frac{1}{\sqrt{2^n}} \sum_{k=0}^{2^n - 1} \left( \sum_{x=0}^{2^n - 1} \alpha_x \omega^{xk} \right) \ket{k}.
\end{equation*}

This shows that the weights for each Fourier basis state are given by:

\begin{equation}\label{fourier_weights}
\tilde{\alpha}_k = \frac{1}{\sqrt{2^n}} \sum_{x=0}^{2^n - 1} \alpha_x \omega^{xk},
\end{equation}

which depend on the interference of the input coefficients $\alpha_x$.

The QFT is a linear transformation, defined by its action on computational basis states $\{ \ket{x} \}$, mapping each $\ket{x}$ to a superposition of $\{ \ket{k} \}$. For $n$-qubits, the transformation is represented by the unitary matrix shown in \autoref{QFTmatrix}. 
\begin{equation}\label{QFTmatrix}
\scalebox{0.97}{$
\frac{1}{\sqrt{2^n}} \begin{bmatrix}
1 & 1 & 1 & \cdots & 1 \\
1 & \omega & \omega^2 & \cdots & \omega^{2^n - 1} \\
1 & \omega^2 & \omega^4 & \cdots & \omega^{2(2^n - 1)} \\
\vdots & \vdots & \vdots & \ddots & \vdots \\
1 & \omega^{2^n - 1} & \omega^{2(2^n - 1)} & \cdots & \omega^{(2^n - 1)(2^n - 1)}
\end{bmatrix}$}
\end{equation}
This matrix can be implemented in a quantum circuit using a sequence of Hadamard gates and controlled-phase gates. For a more elaborate introduction to QFT, see \cite{nielsen2000quantum}.

\paragraph{Boolean Functions as Phase Oracles}

A Boolean function $f(x)$ can be encoded in a quantum algorithm through a phase oracle
\[
U_f\ket{x}=(-1)^{f(x)}\ket{x}.
\]
This is the form used in the main algorithms. One standard way to implement it is to start from the reversible Boolean oracle
\begin{equation}\label{oracle_equation}
O_f\ket{x}\ket{y}=\ket{x}\ket{y\oplus f(x)},
\end{equation}
where $\oplus$ denotes addition modulo $2$. Applying $O_f$, then a $Z$ gate to an ancilla initialized in $\ket{0}$, and then $O_f$ again gives
\[
\ket{x}\ket{0}
\mapsto
\ket{x}\ket{f(x)}
\mapsto
(-1)^{f(x)}\ket{x}\ket{f(x)}
\mapsto
(-1)^{f(x)}\ket{x}\ket{0}.
\]
Thus the ancilla is uncomputed, and the Boolean function is encoded as the phase oracle $U_f$.

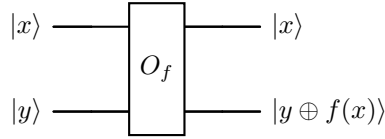
\begin{figure}[H]
    \centering
    \begin{quantikz}
\lstick{$\ket{x}$} & \qw & \gate[2]{O_f} & \qw & \rstick{$\ket{x}$} \qw \\
\lstick{$\ket{y}$} & \qw & \qw & \qw & \rstick{$\ket{y \oplus f(x)}$} \qw
\end{quantikz}
    \caption{Reversible Boolean oracle $O_f$ used to implement the phase oracle $U_f\ket{x}=(-1)^{f(x)}\ket{x}$. The oracle leaves $\ket{x}$ unchanged and maps the ancilla to $\ket{y\oplus f(x)}$.}
    \label{fig:oracle}
\end{figure}

\section{Searching with Quantum Fourier Transform}
\label{sec:searchQFT}

The Quantum Fourier Transform is an efficient method for extracting periodicity and symmetries in quantum states. Typically, QFT comes together with Boolean functions that encode the problem into a quantum system. See \autoref{basics_QFT} for an elementary introduction to QFT and quantum oracles for encoding Boolean functions.

\paragraph{Encoding the Marked Set into Phase Factors}

Let $f(x)$ be a Boolean function that equals 1 for the entries of interest (e.g., all squares in a subset of the integers) and 0 elsewhere. To perform the QFT, the function $f(x)$ must first be encoded in an appropriate form within a quantum state.

Start with a quantum circuit consisting of $n$ registers and an ancillary qubit. All qubits are initialized to the $\ket{0}$ state, resulting in the initial state:

\begin{equation*}
\ket{\psi_0} = \ket{0}^{\otimes n} \otimes \ket{0}.
\end{equation*}

Next, apply Hadamard gates to all $n$ registers to create a uniform superposition over all computational basis states $\ket{x}$:

\begin{equation*}
\ket{\psi_1} = \frac{1}{\sqrt{2^n}} \sum_{x=0}^{2^n - 1} \ket{x} \otimes \ket{0}.
\end{equation*}

Use the reversible Boolean oracle $O_f$ from \autoref{oracle_equation}, which computes $f(x)$ into the ancillary qubit:
$O_f \ket{x} \ket{0} = \ket{x} \ket{f(x)}$. This gives

\begin{equation*}
\ket{\psi_2} = O_f \ket{\psi_1} = \frac{1}{\sqrt{2^n}} \sum_{x=0}^{2^n - 1} \ket{x} \otimes \ket{f(x)}.
\end{equation*}


We then apply a phase conditioned on the ancillary qubit and uncompute by applying $O_f$ again. This restores the ancilla to $\ket{0}$ and leaves $f(x)$ encoded only as the phase oracle $U_f\ket{x}=(-1)^{f(x)}\ket{x}$:

\begin{equation}\label{f(x)_encoded_in_phase}
\ket{\psi_3}
=
\frac{1}{\sqrt{2^n}}
\sum_{x=0}^{2^n - 1}
(-1)^{f(x)}\ket{x}\otimes\ket{0}.
\end{equation}

\paragraph*{Final Step}

The function $f(x)$ is now encoded as phase factors in the computational basis states, with the ancilla uncomputed. To analyze the periodicities in $f(x)$, we apply the QFT to the computational registers. The QFT transformation, acts on the computational basis states as:

\[
QFT \ket{x} = \frac{1}{\sqrt{2^n}} \sum_{k=0}^{2^n - 1} \omega^{xk} \ket{k},
\]
where $\omega = e^{2\pi i / 2^n}$. Applying the QFT to \autoref{f(x)_encoded_in_phase} yields:

\begin{equation*}
QFT \ket{\psi_3} = \frac{1}{\sqrt{2^n}} \sum_{x=0}^{2^n - 1} (-1)^{f(x)} \left( \frac{1}{\sqrt{2^n}} \sum_{k=0}^{2^n - 1} \omega^{xk} \ket{k} \right) \otimes \ket{0}.
\end{equation*}

This can be rewritten as:

\begin{equation*}
\ket{\psi_4} = \frac{1}{2^n} \sum_{k=0}^{2^n - 1} \left( \sum_{x=0}^{2^n - 1} (-1)^{f(x)} \omega^{xk} \right) \ket{k} \otimes \ket{0}.
\end{equation*}

The state $\ket{\psi_4}$ encodes information about the periodicities in $f(x)$. For example, if $f(x)$ is periodic with period $p$, the QFT output will exhibit prominent peaks at $k = m \cdot 2^n / p$, where $m$ is an integer.

\section{Shifted oracles and differences between first terms}
\label{appendix_shifted_oracles}

In the setting of \autoref{theorem_3_n_progressions}, suppose the marked set contains arithmetic progressions
$$
A_j=\{l_j,\ l_j+k,\ l_j+2k,\dots,l_j+(n-1)k\}, \qquad j=1,\dots,n,
$$
all with the same common difference $k$, and let
$$
S=\bigcup_{j=1}^n A_j \subseteq \{1,\dots,B\}.
$$
Let $f:\{1,\dots,B\}\to\{0,1\}$ be the indicator of $S$, and let the phase oracle be
$$
U_f\ket{x}=(-1)^{f(x)}\ket{x}.
$$

For an integer shift $s$, define the shifted oracle by conjugating $U_f$ with the shift $x\mapsto x+s$ on the domain whenever this remains inside $\{1,\dots,B\}$. Indeed, if $T_s$ denotes the shift operator $T_s\ket{x}=\ket{x+s}$, then
$$
T_s^\dagger U_f T_s\ket{x}=(-1)^{f(x+s)}\ket{x},
$$
so conjugation by $T_s$ produces the phase oracle for the shifted indicator $x\mapsto f(x+s)$.

Interfering $U_f$ with this shifted oracle probes the overlap
$$
C(s)=|S\cap(S-s)|.
$$

Composing $U_f$ with the shifted oracle gives
$$
U_fU_f^{(s)}\ket{x}=(-1)^{f(x)+f(x+s)}\ket{x},
$$
so, when this operator is tested on a uniform superposition, e.g., expectation value by the Hadamard-test \cite{lin2022lecturenotesquantumalgorithms,10.1007/978-3-030-35423-7_21}, the resulting signal depends on the overlap $C(s)=|S\cap(S-s)|$.

Now consider two progressions $A_i$ and $A_j$. Their overlap after shifting by $s$ is nonzero exactly when
$$
s=(l_j-l_i)+mk
$$
for some $m\in\{-(n-1),\dots,n-1\}$. Therefore, the peaks of the autocorrelation $C(s)$ occur at values of the form
$$
(l_j-l_i)+mk.
$$
The signal recovers the first terms $l_1,\dots,l_n$ indirectly by the peaks arranged in packets centered at the pairwise differences $l_j-l_i$, with known offsets $mk$.

Since $k$ is already the quantity targeted by the QFT step, shifted-oracle interference can be used to extract the difference set
$$
\{\,l_j-l_i : i,j\in\{1,\dots,n\}\,\}.
$$

\section{Proofs of theorems, propositions \& lemmas}\label{proofs}

\begin{proof}[\textbf{Proof of \autoref{triangles}}]
For any constant $s\neq M/3$, the lines $\mathcal{L}^{[x,s]}$, $\mathcal{L}^{[y,s]}$, and $\mathcal{L}^{[z,s]}$ intersect in the plane $x+y+z=M$ to form an equilateral triangle. By \autoref{parallels}, any two points in this plane that share a coordinate determine a line belonging to one of the three families of parallel lines. In particular, any line of the form $\mathcal{L}^{[x,k_1]}$, $\mathcal{L}^{[y,k_2]}$, or $\mathcal{L}^{[z,k_3]}$ is parallel to $\mathcal{L}^{[x,s]}$, $\mathcal{L}^{[y,s]}$, or $\mathcal{L}^{[z,s]}$, respectively, independently of the values of $k_1$, $k_2$, and $k_3$.
\end{proof}

\begin{proof}[\textbf{Proof of \autoref{proposition1}}]
The operator $\mathcal{T}_{(i,j)}$ acts by swapping two coordinates of each point, and therefore permutes the three families of parallel lines that determine the hexagon (see \autoref{fig:hexagon}). Since coordinate transpositions preserve distances and angles, they preserve regularity; hence the transformed points $p'_1,\ldots,p'_6$ again form a regular hexagon.

It remains to show that $p'_0$ is its barycenter. In the original configuration, the barycenter $p_0$ is characterized by the same coordinate relations with the six vertices, and these relations are preserved under $\mathcal{T}{(i,j)}$, since the map only permutes coordinates. Therefore, $\mathcal{T}_{(i,j)}({p_0})={p'_0}$ is the barycenter of the transformed hexagon.
\end{proof}

\begin{proof}[\textbf{Proof of \autoref{middle=n/3}}]
The system of Diophantine equations \autoref{dioph_system} corresponding to the magic square  can be written in matrix form as follows:

\begin{equation}\label{system_matrix_form}
\scalebox{0.9}{$
\begin{bmatrix} 
1 & 1 & 1 & 0 & 0 & 0 & 0 & 0 & 0\\
0 & 0 & 0 & 1 & 1 & 1 & 0 & 0 & 0\\
0 & 0 & 0 & 0 & 0 & 0 & 1 & 1 & 1\\
1 & 0 & 0 & 1 & 0 & 0 & 1 & 0 & 0\\
0 & 1 & 0 & 0 & 1 & 0 & 0 & 1 & 0\\
0 & 0 & 1 & 0 & 0 & 1 & 0 & 0 & 1\\
0 & 0 & 1 & 0 & 1 & 0 & 1 & 0 & 0\\
1 & 0 & 0 & 0 & 1 & 0 & 0 & 0 & 1

\end{bmatrix} 
\begin{bmatrix} x_{a}\\ x_{d}\\ x_{g}\\ x_{b}\\ x_{e}\\ x_{h}\\ x_{c}\\ x_{f}\\ x_{i}\end{bmatrix}
= \begin{bmatrix} M\\ M\\ M\\ M\\ M\\ M\\ M\\ M\end{bmatrix}
$}
\end{equation}

The general solution of the system over $\mathbb{R}$ can be expressed as the sum of a particular solution and the nullspace of the associated matrix, as described in \cite{strang2011}. Recall that the nullspace (or kernel) of a matrix $A$, 
denoted by
$\left\{ \vec{x} : A\vec{x} = \vec{0}\right\}$ represents the set of all solutions to the homogeneous system. To compute the nullspace, we determine the reduced row echelon form of the augmented matrix associated with the system. This process simplifies the matrix while preserving its solution set, allowing us to identify a basis for the nullspace and construct the general solution.

\begin{equation*}
\scalebox{0.8}{$
\left[
\begin{array}{@{}*{9}{c}@{\hspace{4pt}}|@{\hspace{4pt}}c@{}}
1 & 1 & 1 & 0 & 0 & 0 & 0 & 0 & 0 & 0 \\
0 & 0 & 0 & 1 & 1 & 1 & 0 & 0 & 0 & 0 \\
0 & 0 & 0 & 0 & 0 & 0 & 1 & 1 & 1 & 0 \\
1 & 0 & 0 & 1 & 0 & 0 & 1 & 0 & 0 & 0 \\
0 & 1 & 0 & 0 & 1 & 0 & 0 & 1 & 0 & 0 \\
0 & 0 & 1 & 0 & 0 & 1 & 0 & 0 & 1 & 0 \\
0 & 0 & 1 & 0 & 1 & 0 & 1 & 0 & 0 & 0 \\
1 & 0 & 0 & 0 & 1 & 0 & 0 & 0 & 1 & 0 \\
\end{array}
\right]
\!\!\sim\!\!
\left[
\begin{array}{@{}*{9}{c}@{\hspace{4pt}}|@{\hspace{4pt}}c@{}}
1 & 0 & 0 & 0 & 0 & 0 & 0 & 0 & 1 & 0 \\
0 & 1 & 0 & 0 & 0 & 0 & 0 & 1 & 0 & 0 \\
0 & 0 & 1 & 0 & 0 & 0 & 0 & -1 & -1 & 0 \\
0 & 0 & 0 & 1 & 0 & 0 & 0 & -1 & -2 & 0 \\
0 & 0 & 0 & 0 & 1 & 0 & 0 & 0 & 0 & 0 \\
0 & 0 & 0 & 0 & 0 & 1 & 0 & 1 & 2 & 0 \\
0 & 0 & 0 & 0 & 0 & 0 & 1 & 1 & 1 & 0 \\
0 & 0 & 0 & 0 & 0 & 0 & 0 & 0 & 0 & 0 \\
\end{array}
\right]
$}
\end{equation*}

\noindent Thus, the nullspace is given by vectors of the form:

\begin{equation*}
\scalebox{0.9}{$
\begin{bmatrix}  -x_i\\-x_f \\x_f +x_i\\x_f +2x_i\\ 0\\-x_f - 2x_i\\-x_f - x_i\\ x_f \\ x_i \end{bmatrix}
= x_f \begin{bmatrix}  0\\-1 \\1\\1\\ 0\\-1\\-1\\1\\0 \end{bmatrix} + x_i \begin{bmatrix}  -1\\ 0 \\1\\2\\ 0\\-2\\-1\\0 \\1 \end{bmatrix}
$}
\end{equation*}

Because every row of \autoref{system_matrix_form} has exactly three $1$s, an evident solution to \autoref{system_matrix_form} (over $\mathbb{R}$) is a vector with $\frac{M}{3}$ in all of its entries. Therefore, the general solution is: 

\begin{equation}\label{zeros_in_middle}
\scalebox{0.9}{$
x_f \begin{bmatrix}  0\\-1 \\1\\1\\ 0\\-1\\-1\\1\\0 \end{bmatrix} + x_i \begin{bmatrix}  -1\\ 0 \\1\\2\\ 0\\-2\\-1\\0 \\1 \end{bmatrix} + 
\begin{bmatrix}
M/3 \\
M/3 \\
M/3 \\
M/3 \\
M/3 \\
M/3 \\
M/3 \\
M/3 \\
M/3
\end{bmatrix}
$}
\end{equation}

Observe that the middle entries of the three vectors above are the entries that correspond to $x_e$. They have the values $0$, $0$, and $\frac{M}{3}$, respectively. Therefore, $x_e = \frac{M}{3}$.  
\end{proof}

\begin{proof}[\textbf{Proof of \autoref{trace_of_diag}}]
Let $j=\left(\frac{M}{3},\frac{M}{3},\frac{M}{3}\right)$, where $M$ is the magic constant of the square in \autoref{fig:3by3_letters}. This point is fixed by every transposition in $\mathscr{T}$ and lies in the plane $x+y+z=M$.

For any point $(a,b,c)$, its Euclidean distance to $j$ is
$\sqrt{\left(\frac{M}{3}-a\right)^2+\left(\frac{M}{3}-b\right)^2+\left(\frac{M}{3}-c\right)^2}$.
Since each operator $\mathcal{T}_{(i,j)}\in\mathscr{T}$ merely permutes the coordinates, this distance is invariant under the action of $\mathscr{T}$. Hence the trajectory of $(a,b,c)$ lies on a sphere centered at $j$. Because all points in the trajectory also satisfy $x+y+z=M$, they lie on the circle obtained by intersecting this sphere with the plane $x+y+z=M$.

Now suppose that $p$ is a triplet corresponding to a diagonal of the magic square. By \autoref{middle=n/3}, the middle entry is $M/3$, so every such triplet contains $M/3$ as one of its coordinates. Therefore every point in the trajectory of $p$ shares one coordinate with $j$. By \autoref{two_points_common_line}, these points lie on lines from the three families of parallel lines through $j$, and by \autoref{triangles} their arrangement is hexagonal. Thus, the six points in the trajectory form a regular hexagon with barycenter $j$.
\end{proof}

\begin{proof}[\textbf{Proof of \autoref{pattern_theorem}}]
The claim follows from the left panel of \autoref{fig:cube_projection} together with the configuration in \autoref{fig:pattern all lines}. The latter exhibits a ``hexagon of hexagons'': a large hexagon whose vertices are the barycenters of six smaller congruent hexagons.

\begin{itemize}
\item \textit{Projecting a hexagon yields a three-term arithmetic progression.}
By the left panel of \autoref{fig:cube_projection}, projecting the vertices and barycenter of a regular hexagon onto a coordinate axis produces three points in arithmetic progression.

\item \textit{Projecting the barycenters yields another three-term arithmetic progression.}
The barycenters of the six small hexagons form a larger hexagon. Projecting these barycenters onto a coordinate axis therefore again produces three points in arithmetic progression.

\item \textit{The full configuration yields three nested progressions.}
Projecting all points of the six small hexagons onto the same axis, each hexagon contributes a three-term arithmetic progression. Since the six hexagons are congruent and aligned as in \autoref{fig:pattern all lines}, these progressions are themselves translates of one another and overlap to give rise to three distinct arithmetic progressions of length $3$. Their middle terms, corresponding to the barycenters, also form an arithmetic progression.
\end{itemize}

Thus, passing from a single hexagon to the hexagon of hexagons replaces each term of one three-term arithmetic progression by another three-term arithmetic progression.

The right panel of \autoref{fig:cube_projection} illustrates the projection of the points in \autoref{fig:pattern all lines}. Any other valid configuration is qualitatively the same, by \autoref{corollary_cube_projection} and the previous results establishing the existence and alignment of the relevant hexagons.
\end{proof}

\begin{proof}[\textbf{Proof of \autoref{theorem_weighted}}] (Sketch)
The proof is a variation of that of \autoref{pattern_theorem}, combined with Thales's proportionality theorem \cite{coxeter1969geometry}. The main new feature is that the weighted plane is tilted relative to the standard plane $x+y+z=M$.

To simplify the argument, we first work on the plane $x+y+z=M$ and then project the resulting configuration onto the weighted plane $w_xx+w_yy+w_zz=M$. The lines $\mathcal{L}^{[j,k]}$, defined on the standard plane by intersecting it with planes perpendicular to the coordinate axes, have analogous counterparts on the weighted plane. Under projection, they again produce a periodic pattern analogous to \autoref{fig:pattern all lines}, but the three families of parallel lines are no longer affected uniformly: depending on the weights $w_x,w_y,w_z$, some directions are stretched while others are contracted.

Thales's proportionality theorem then shows that the pattern observed on the $i$th axis is stretched or contracted relative to its most stretched occurrence by a factor determined by the weight $w_i$. In the special case $w_x=w_y=w_z$, the weighted plane reduces to the standard plane $x+y+z=M$, so all projected patterns coincide. Thus, \autoref{pattern_theorem} is the equal-weight case of the present result.
\end{proof}

\begin{proof}[\textbf{Proof of \autoref{bound_non_existence}}]
Consider a $3\times3$ magic square whose largest entry is $t^z$ and whose second largest entry is $(t-1)^{z-1}$. Normalize all entries by dividing by $t^z$, so that the largest entry becomes $1$.

The normalized gap between the two largest entries is then
$\frac{t^z-(t-1)^{z-1}}{t^z}$.
As $t\to\infty$,
$$\lim_{t\to\infty}\frac{t^z-(t-1)^{z-1}}{t^z}=1.$$
Thus, for sufficiently large $t$, the gap between the largest and second largest normalized entries becomes arbitrarily close to $1$.

By \autoref{pattern_theorem}, a solution must give rise to three arithmetic progressions of length $3$. This becomes impossible once the normalized gap is at least $\frac{1}{3}$, since \autoref{pattern_theorem} requires three arithmetic progressions of length $3$, and such a configuration cannot lie in an interval of length $1$ with common difference at least $\frac{1}{3}$. The threshold is determined by
$\frac{t^z-(t-1)^{z-1}}{t^z}=\frac{1}{3}$.

Therefore, there exists $t_0$ such that for all $t\ge t_0$ no such magic square exists. Equivalently, any such magic square must have all entries $\le U$, where $U=(t_0+1)^z$.
\end{proof}

\begin{proof}[\textbf{Proof of \autoref{theorem_3_n_progressions}}]
Let $[n^2]=\{1,2,\dots,n^2\}$. For each residue $r\in\{0,\dots,n-1\}$, define
$
\mathrm{Fiber}(r)=\{m\in\mathbb{Z}: m\bmod n=r\}.
$
Each fiber contains exactly $n$ elements of $[n^2]$. We partition $[n^2]$ into $n$ sets by taking the smallest remaining element from each fiber at each step:
\begin{align*}
C_1 &= \bigcup_{r=0}^{n-1}\min(\mathrm{Fiber}(r)),\\
C_2 &= \bigcup_{r=0}^{n-1}\min(\mathrm{Fiber}(r)\setminus C_1),\\
C_3 &= \bigcup_{r=0}^{n-1}\min\!\bigl(\mathrm{Fiber}(r)\setminus (C_1\cup C_2)\bigr),\\
&\ \vdots\\
C_n &= \bigcup_{r=0}^{n-1}\min\!\bigl(\mathrm{Fiber}(r)\setminus \textstyle\bigcup_{i=1}^{n-1}C_i\bigr).
\end{align*}

Now place the integers $\max(C_1),\dots,\max(C_n)$ into a diagonal Latin square $Lat_1$ of order $n$. A diagonal Latin square is a Latin square whose two diagonals are transversals. Such squares exist for every $n>3$, and orthogonal diagonal Latin squares exist for every $n\neq 2,3,6$ \cite{c8ff8faf0b5e412fb9dcf5933f672e20,hilton1973double,Brown1993}. Let $Lat_2$ be an orthogonal diagonal Latin square of the same order.

Superimposing $Lat_1$ and $Lat_2$ produces a square $Lat_3$ whose entries are distinct pairs of the form $(\max(C_i),j)$, where $i\in[n]$ and $j\in\{0,\dots,n-1\}$. Applying the bijection $F(x,y)=x-y$ to each entry yields a square $Lat_4$ filled with the integers in $[n^2]$. By construction, $Lat_4$ is a normal magic square. This proves the claim for consecutive integers.

The same argument applies to any arithmetic progression of length $n^2$, since translating all entries by a fixed integer preserves the magic-square property.

A further variation is also possible. The construction depends on the $n$ values used to fill $Lat_1$. If some of these values are shifted by fixed amounts, the Latin-square structure still preserves the row, column, and diagonal sums. One must only ensure that, after applying $F(x,y)$, the resulting entries remain distinct.

Such a shift changes exactly $n$ entries of the resulting magic square, and these entries form an arithmetic progression of length $n$. Indeed, for fixed $d\in[n]$, the entries
$
(\max(C_d),1),(\max(C_d),2),\dots,(\max(C_d),n)
$
in $Lat_3$ are sent by $F$ to
$
D-1,D-2,\dots,D-n,
$
where $D=\max(C_d)$. The same holds for every $d\in[n]$. Therefore the construction extends from consecutive integers to patterns built from $n$ arithmetic progressions of length $n$, which proves the theorem.
\end{proof}

\begin{proof}[\textbf{Proof of \autoref{byShors}}]
By \autoref{pattern_theorem}, a solution corresponds to a pattern of $9$ integers of the form
\begin{multline}
l,\ l+k,\ l+2k,\ l+K,\ l+k+K,\ l+2k+K,\ l+2K,\ l+k+2K,\ l+2k+2K,
\end{multline}
where $l,k,K\in\mathbb{N}$ determine the pattern. For a magic square of squares, the shifting constant must be a square, so
$l=s^2$ for some $s\in\mathbb{N}$.
Assume that at least one defining distance is a square, i.e.,
$k=t^2$ or $K=u^2$ for some $t,u\in\mathbb{N}$.

If $k=t^2$, then
\[
l+k=s^2+t^2.
\]
If $K=u^2$, then
\[
l+K=s^2+u^2.
\]
Hence, in either case, the existence of a solution forces at least one integer in the pattern to satisfy
\[
x^2+y^2=z
\]
within the marked set.

By the sum of two squares theorem \cite{dudley1969}, an integer $z>1$ satisfies $x^2+y^2=z$ with $x,y\in\mathbb{Z}$ if and only if its prime factorization contains no prime $p\equiv 3 \bmod 4$ to an odd power \cite{heathbrown1984,hardy1938}. Classically, the bottleneck is factoring $z$. Shor's algorithm factors such integers efficiently, enabling a hybrid classical--quantum procedure to test this condition over the marked set and thereby certify the absence of solutions.
\end{proof}

\begin{proof}[\textbf{Proof of \autoref{byShors2}}]
The proof is analogous to that of \autoref{byShors}. The pattern is determined by the shifting constants $l_1,l_2,\ldots,l_n$ together with the defining distances $k$ and $K$. In the case of magic squares of squares, the relevant terms are again constrained by equations of the form
\[
x^2+y^2=z,
\]
so the same combination of Shor's algorithm and the sum of two squares theorem applies.

Unlike \autoref{byShors}, however, this proposition rules out only solutions of the specific form described in \autoref{theorem_3_n_progressions}, rather than all possible solutions in the marked set.
\end{proof}

\begin{proof}[\textbf{Proof of \autoref{Recovering_D_shift}}]

Using notation consistent with \autoref{appendix_shifted_oracles}:

Let, $$ P=\bigcup_{j=1}^n A_j,\qquad A_j=\{l_j,l_j+k,\dots,l_j+(n-1)k\},$$
where the starting points satisfy
$$l_j=l_1+(j-1)D$$
for some spacing $D$. Assume that
$$D>2(n-1)k.$$
Define the autocorrelation of the clean pattern $P$ by
$$C_P(s)=|P\cap(P-s)|.$$

Consider two progressions $A_i$ and $A_j$. An overlap at shift $s$ occurs precisely when $l_i+uk = l_j+vk-s$, i.e.,
$$s=(l_j-l_i)+(u-v)k$$
for some $u,v\in\{0,\dots,n-1\}$. Since
$$l_j-l_i=(j-i)D,$$
every nonzero contribution to $C_P(s)$ occurs at a shift of the form
$$s=qD+mk,$$
where $q,m\in\{-(n-1),\dots,n-1\}$.

For a fixed value of $q$, there are exactly $n-|q|$ pairs $(i,j)$ with $j-i=q$. Similarly, for a fixed value of $m$, there are exactly $n-|m|$ pairs $(u,v)$ with $u-v=m$. Hence
$$C_P(qD+mk)=(n-|q|)(n-|m|).$$
Therefore the autocorrelation of the clean pattern consists of triangular peaks centered at
$$0,\pm D,\pm 2D,\dots,\pm(n-1)D.$$ (see \autoref{Autocorrelation_figure}).
The peak at $0$ is the trivial self-overlap peak, while the off-center peaks encode the spacing $D$ between the starting points.

\begin{figure}[H]
    \centering
    \includegraphics[width=0.8\linewidth]{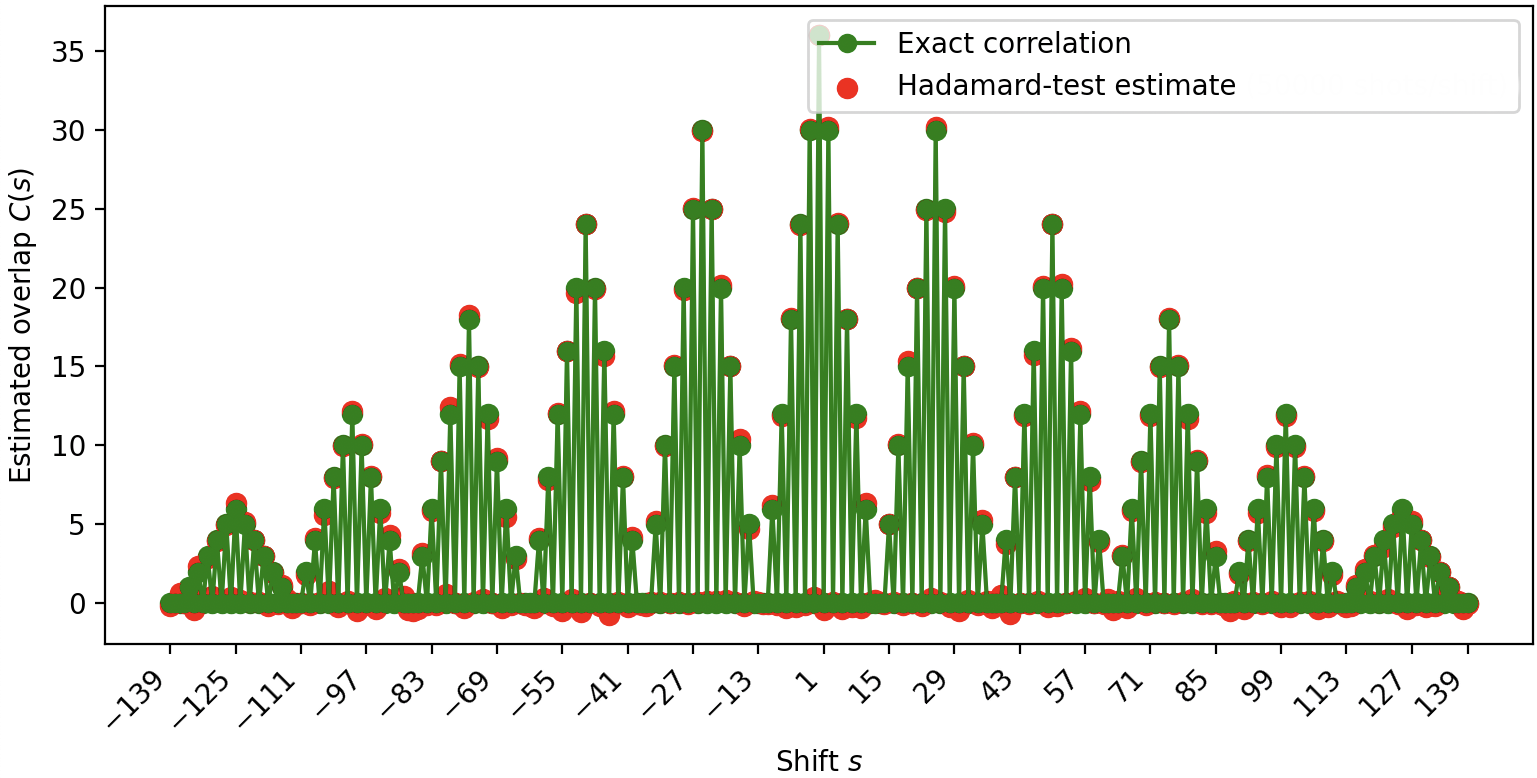}
    \caption{Autocorrelation for 6 arithmetic progressions of 6 terms, distance between arithmetic progressions D=25, string's size 256 with 0.1 noise. Note that, due to the triangular structure of the autocorrelation peaks, evaluating two shifts on each side of a peak is enough to infer its center.}\label{Autocorrelation_figure}
    \label{fig:spectrum1}
\end{figure}

Now let $N$ denote the set of additional marked points, so that
$$S=P\cup N.$$
Define
$$C_S(s)=|S\cap(S-s)|=C_P(s)+\Delta(s),$$
where $\Delta(s)$ denotes the perturbation produced by $N$. Assume that
$$|\Delta(s)|\le \eta$$
for all relevant shifts $s$, where $\eta$ is smaller than half of the one-step drop along the off-center peaks of $C_P$. Under this assumption (or less tight assumption for larger steps), the off-center peaks remain separated and retain their local shape, so their centers can still be identified from local evaluations of the signal.

It follows that the same tracing procedure as in the clean case still applies: once a sampled shift falls inside an off-center peak, a constant number of nearby evaluations determines its center. Since the off-center peaks are equally spaced, identifying one of them determines the spacing $D$, and hence the full starting-point pattern. The number of tested shifts is polynomial in $n$, so the reconstruction is polynomial-time.

For each fixed shift $s$, the Hadamard-test estimates a $\{\pm1\}$-valued quantity. Therefore estimating the normalized signal to fixed additive accuracy $\varepsilon$ requires $O(1/\varepsilon^2)$ shots, independently of the length of the binary string.
\end{proof}

\section{Examples}\label{examples}

\begin{example}\label{example_4_by_4}
The following examples illustrate the construction of magic squares from \autoref{theorem_3_n_progressions}. We begin with a normal magic square of order $4$ with entries from $1$ to $16$ (\autoref{fig:4by4_magic_square}).

\renewcommand{\arraystretch}{1.6} 
\begin{table}[H]
    \centering
    \begin{tabular}{|c|c|c|c|}
        \cline{1-4} 
        7 & 12 & 1 & 14 \\
        \cline{1-4}
        2 & 13 & 8 & 11 \\
        \cline{1-4}
        16 & 3 & 10 & 5 \\
        \cline{1-4}
        9 & 6 & 15 & 4 \\
        \cline{1-4}
    \end{tabular}
\vspace{0.5cm}        
    \caption{\normalfont A normal magic square of order $4$.}
    \label{fig:4by4_magic_square}
\end{table}

In this square, the values $\max C_1,\max C_2,\max C_3,\max C_4$ used in $Lat_1$ are $4,8,12,16$. Increasing them by $0$, $2$, $4$, and $9$, respectively, yields the magic square in \autoref{fig:4by4_magic_square_new}.

\renewcommand{\arraystretch}{1.6}
\begin{table}[H]
    \centering
    \begin{tabular}{|c|c|c|c|}
        \cline{1-4} 
        9  & 16 &  1 & 23 \\
        \cline{1-4}
        2  & 22 & 10 & 15 \\
        \cline{1-4}
        25 &  3 & 14 &  7 \\
        \cline{1-4}
        13 &  8 & 24 &  4 \\
        \cline{1-4}
    \end{tabular}
\vspace{0.5cm}        
    \caption{\normalfont An order $4$ magic square obtained by shifting the entries used in $Lat_1$.}
    \label{fig:4by4_magic_square_new}
\end{table}

To illustrate that the progressions may also have common difference greater than $1$, consider the following order-$4$ magic square, constructed from four arithmetic progressions with common difference $2$ (\autoref{fig:4by4_magic_square_step_2}).

\renewcommand{\arraystretch}{1.6}
\begin{table}[H]
    \centering
    \begin{tabular}{|c|c|c|c| }
        \cline{1-4} 
        17  & 31 &  1 & 39 \\
        \cline{1-4}
        3  & 37 & 19 & 29 \\
        \cline{1-4}
        43 &  5 & 27 &  13 \\
        \cline{1-4}
        25 &  15 & 41 &  7 \\
        \cline{1-4}
    \end{tabular}
\vspace{0.5cm}    
    \caption{\normalfont An order $4$ magic square obtained from four arithmetic progressions with common difference $2$, constructed using \autoref{theorem_3_n_progressions}.}
    \label{fig:4by4_magic_square_step_2}
\end{table}
\end{example}

\end{document}